\documentclass[showpacs,superscriptaddress,aps,pra,twocolumn]{revtex4-2}

\UseRawInputEncoding

\usepackage{graphicx}

\usepackage{bm}
\usepackage{amsmath, amsthm, amssymb,amsfonts,mathbbol,amstext,mathrsfs,mathtools}

\usepackage[colorlinks=true,linktocpage=true]{hyperref}
\hypersetup{
    allcolors  = {blue},
}
\usepackage[normalem]{ulem}

\usepackage{dsfont}

\newtheorem{theorem}{Theorem}

\newtheorem{corollary}{Corollary}

\newtheorem{lemma}{Lemma}

\theoremstyle{definition}
\newtheorem{definition}{Definition}

\renewcommand{\r}{\rightarrow}
\newcommand{\be}{\begin{equation}}
\newcommand{\ee}{\end{equation}}
\newcommand{\C}{\mathcal}
\newcommand{\tb}{\textbf}
\newcommand{\ti}{\textit}

\usepackage{cancel}
\usepackage{soul}
\usepackage{color}

\usepackage{tikz}
\tikzstyle{copy}=[fill=white, draw=black, shape=circle]
\tikzstyle{none}=[inner sep=0pt]
\tikzstyle{new}=[-, thick, draw={rgb,255: red,81; green,41; blue,241}]
\tikzstyle{new_dashed}=[-, thick, dashed, draw={rgb,255: red,81; green,41; blue,241}]
\tikzstyle{alice}=[-, fill={rgb,255: red,124; green,231; blue,255}, draw={rgb,255: red,12; green,60; blue,216}]
\tikzstyle{bob}=[-, fill={rgb,255: red,255; green,123; blue,125}, draw={rgb,255: red,171; green,18; blue,21}]
\tikzstyle{alice_dashed}=[-, dashed, draw={rgb,255: red,74; green,201; blue,255}]
\tikzstyle{bob_dashed}=[-, dashed, draw={rgb,255: red,255; green,123; blue,125}]
\pgfdeclarelayer{edgelayer}
\pgfdeclarelayer{nodelayer}
\pgfsetlayers{edgelayer,nodelayer,main}

\begin{document}

\newcommand{\pisa}{Department of Physics ``E. Fermi'', University of Pisa, Largo B. Pontecorvo 3, 56127 Pisa, Italy} 
\newcommand{\pisashort}{Department of Physics ``E. Fermi'', University of Pisa, Pisa, Italy} 

\newcommand{\cfum}{Centro de F\'{i}sica, Universidade do Minho, Campus de Gualtar, 4710-057 Braga, Portugal} 
\newcommand{\cfumshort}{Centro de F\'{i}sica, Universidade do Minho, Braga, Portugal} 

\newcommand{\inl}{INL -- International Iberian Nanotechnology Laboratory, Av. Mestre Jos\'{e} Veiga s/n, 4715-330 Braga, Portugal} 
\newcommand{\inlshort}{INL -- International Iberian Nanotechnology Laboratory, Braga, Portugal} 


\title{Convexity of noncontextual wirings and how they order the set of correlations}

\author{Tiago Santos}
\email[These authors contributed equally\\ Corresponding author: ]{rafael.wagner@inl.int}
\affiliation{Department of Mathematical Physics, Institute of Physics,
University of S\~ao Paulo, R. do Mat\~ao 1371, S\~ao Paulo 05508-090,
SP, Brazil}

\author{Rafael Wagner}
\email[These authors contributed equally\\ Corresponding author: ]{rafael.wagner@inl.int}
\affiliation{\inlshort}
\affiliation{\cfumshort}
\affiliation{\pisashort}

\author{B\'{a}rbara Amaral}
\affiliation{Department of Mathematical Physics, Institute of Physics,
University of S\~ao Paulo, R. do Mat\~ao 1371, S\~ao Paulo 05508-090,
SP, Brazil}

\begin{abstract}
The resource theory of contextuality considers resourceful objects to be probabilistic data-tables, known as correlations or behaviors, that fail to have an explanation in terms of Kochen-Specker noncontextual models. In this work, we advance this resource theory, considering free operations to be noncontextual wirings (NCW). We show that all such wirings form a convex set. When restricted to Bell scenarios, we show that such wirings are not equivalent to local operations assisted by a common source of classical shared randomness (LOSR). The set of all NCW operations contains LOSR, but is strictly larger. We also prove several elementary facts about how different resources can be converted via NCW. As a concrete example, we show that there are pairs of behaviors that cannot be converted one into the other using NCW. Since resource conversion mathematically induces a pre-order over the set of all behaviors, our results reveal the intricate ordering induced by NCW in scenarios beyond Bell scenarios.
\end{abstract}

\maketitle

\tableofcontents

\section{Introduction}
\label{sec:intro}

Kochen--Specker (KS) contextuality~\cite{budroni2021quantum,kochen1975problem} is a well-known resource for quantum computation~\cite{howard2014contextuality,raussendorf2013contextuality,abramsky2017contextual,bermejo2017contextuality,raussendorf2020phase,oestereich2017reliable,bravyi2018quantum,bravyi2020quantum} and quantum communication~\cite{gupta2023quantum,saha2019stateindependent}. This resource is formally defined in terms of so-called compatibility scenarios (a.~k.~a. measurement scenarios). When KS-noncontextual models can explain the probability distributions from experimental data, the aforementioned applications lack any form of quantum advantage. Beyond these applications, KS-contextuality has also been shown useful for self-testing~\cite{xu2023stateindependent,bharti2019robustself,hu2023self,saha2020sumofsquares}, dimension witnessing~\cite{vidick2011ignorance,guhne2014bounding,ray2021graph}, and even in quantum foundations, for analyzing the  quantum Cheshire cat protocol~\cite{hance2023contextuality}, or allowing the introduction of new extended Wigner's friend arguments~\cite{szangolies2023quantum,walleghem2024extended}.

The term `resource' alone is vague and unclear. Within the resource theories framework~\cite{fritz2017resource,chitambar2019quantum} such terminology has a simple and precise meaning. Consider a set $\mathcal{U}$ and a proper subset denoted as $\mathcal{F}$. The term resource refers to some generic property, that can be of quantum states~\cite{chitambar2019quantum,gour2024resources}, quantum channels~\cite{gour2020dynamical,theurer2020quantifying}, probability distributions (as the ones we will consider), or yet other properties such as nonclassicality of common-cause~\cite{wolfe2020quantifying}, satisfied by elements in $\mathcal{U}\setminus \mathcal{F}$ but not by elements in $\mathcal{F}$. The intuition behind the word `resource' is that access to $\mathcal{U}\setminus \mathcal{F}$ allows one to do `more' than if only access to  $\mathcal{F}$ is granted. Elements of $\mathcal{F}$ are called `free'. For our case of interest, KS-contextuality, a property held by behaviors in compatibility scenarios, will be our defining notion of resource. Free resources will be behaviors having a KS-noncontextual model. 

In the resource theories framework, given that we have a notion of free objects determined by elements from the set $\mathcal{F}$, we can also have a notion of free operations as those that, at least, preserve this set. In other words, mappings $T:\mathcal{U} \to \mathcal{U}$ such that $T(\mathcal{F}) \subseteq \mathcal{F}$. Different works have considered different choices of free operations, meaning in our case those that do not create KS-contextual correlations from KS-noncontextual ones~\cite{horodecki2023rank,horodecki2015axiomatic,grudka2014quantifying,abramsky2019comonadic,karvonen2021neither,barbosa2023closing}.  \emph{Noncontextual wirings}~\cite{amaral2018noncontextual,amaral2019resource} is one such choice of free operations. These wirings can be intuitively described as the composition of pre- and post-processing, combined with adaptivity between KS-noncontextual correlations. They form a useful working subset of all possible free operations for a resource theory of KS-noncontextuality. However, there are currently two outstanding issues regarding the definition of noncontextual wirings. 

The first issue is the relation between NCW with the set of free operations known as \emph{local operations with shared randomness (LOSR)}, which has two non-equivalent formulations, one investigated by Gallego and Aolita~\cite{gallego2017nonlocality}, and the other introduced by Vicente~\cite{deVicente2014LOSR},  Geller and Piani~\cite{geller2014quantifying}. That these two formulations are not equivalent was pointed out by Wolfe \emph{et. al.}~\cite{wolfe2020quantifying}. In light of this, noncontextual wirings have been claimed to be equivalent to Gallego and Aolita's description of LOSR, which we will term as $\mathrm{LOSR}^\circ$, when we restrict wirings to act on compatibility scenarios mathematically equivalent to Bell scenarios. In this setting, it is unclear which relation (if any) exists between NCW and the Vicente-Geller-Piani description of LOSR, which we will term as $\mathrm{LOSR}^\bullet$. Beyond that, it has also been pointed out recently by Karvonen~\cite{karvonen2021neither} that general free operations for noncontextuality should also allow for wirings between different parties when applied to compatibility scenarios mathematically equivalent to Bell scenarios. This also suggests some further revision of the definition of NCW, or some previous results, as when restricted to Bell scenarios, this feature pointed out by Karvonen would imply that NCW is \emph{inequivalent} even to $\mathrm{LOSR}^\circ$, contrary to what was claimed in Refs.~\cite{amaral2018noncontextual,amaral2019resource}.

The second issue concerns the \emph{convexity} of the set of NCW operations. As pointed out by Wolfe~\emph{et.~al.} both the sets $\mathrm{LOSR}^\circ$ and NCW satisfy that the pre- and post-processings allowed have \emph{independent} sources of classical randomness. This is the reason why $\mathrm{LOSR}^\circ$ is a non-convex subset of $\mathrm{LOSR}^\bullet$, the latter allowing for the same common source of randomness between pre- and post-processing behaviors. This discussion suggests either that NCW operations should form a non-convex subset, or that we should revise its definition to allow for a common source of randomness. However, this is not entirely clear since, as pointed out by Ref.~\cite[Remark 28]{barbosa2023closing}, having access to a common source of shared randomness or to some form of adaptivity on the free operations results in similar expressive power regarding the possible free transformations. Anything one can do, operationally, the other can also do, and vice-versa.  

In this work, we resolve both the issues raised above. First, we show that both $\mathrm{LOSR}^\circ$ and $\mathrm{LOSR}^\bullet$ are proper subsets of NCW when restricted to compatibility scenarios mathematically equivalent to Bell scenarios. Hence, for such scenarios, any definition of LOSR is not equivalent to NCW. We also show that type-dependent NCW operations always form convex sets (in fact, convex polytopes) by showing that, due to the adaptivity, it is always possible to re-write the definition of a noncontextual wiring as if both pre- and post-processing boxes have a common source of shared randomness. 

Beyond the relationship between LOSR and NCW, we also prove structural results on how NCW operations can act on behaviors in compatibility scenarios. One aspect of the expressive power of free operations is their ability to relate different resources. If there exists a free operation $T$ such that $T(A) = B$ we write that $A \to B$. This free arrow defines a pre-order (a relation that is both reflexive and transitive) on the set of all possible objects in a resource theory. As part of the characterization of noncontextual wirings, we investigate some specific properties of how different objects can be transformed one into another. 
\begin{enumerate}
    \item We show that there are objects $B_1$ and $B_2$ such that neither $B_1$ can be transformed using NCW operations into $B_2$, and vice-versa. When this happens we denote $B_1 \nleftrightarrow B_2$, and say that the two behaviors are \emph{incomparable}.
    \item We also show that there are uncountably infinite sets where every pair of elements satisfy this property, i.e., \emph{every} pair of objects is incomparable.
    \item We prove the existence of triplets $B_1,B_2,B_3$ satisfying that $B_1 \nleftrightarrow B_2 \nleftrightarrow B_3$ but $B_3 \to B_1$.
    \item We also show there are infinite chains of behaviors $B_1 \to B_2 \to B_3 \to \dots $ comparable in sequence. 
    \item Finally, we show that there are infinitely many equivalence classes $[B]:= \{A: B \to A \text{ and }A \to B\}$, within the interval $B_1 \to B \to B_2$ marked by some pair of elements $B_1$ and $B_2$. 
\end{enumerate}
We refer to these structural properties as \emph{global comparability properties}. These results are perhaps unsurprising, as these properties were also shown in Ref.~\cite{wolfe2020quantifying} to be true for correlations violating the Clauser-Horne-Shimony-Holt (CHSH) inequalities~\cite{clauser1969proposed}, and when considering $\mathrm{LOSR}^\bullet$ as the free operations. Here we improve these findings by showing that for compatibility scenarios known as the $n$-cycle scenarios~\cite{araujo2013all}, for every integer $n\geq 3$, we can find families of behaviors satisfying all the global properties listed above, when considering NCW as the free operations.

The structure of our work is as follows. We start reviewing the theoretical concepts necessary to present our main results. In Sec.~\ref{sec: resource} we introduce the resource theories framework, with a presentation focusing on providing physical intuition to the mathematical descriptions introduced. Our presentation in this section is aimed towards researchers less familiar with mathematical topics, and motivations behind using this framework, implying that we have filled this section with examples and applications of different resource theoretic notions. Sec.~\ref{sec: contextuality} reviews KS-noncontextuality and Sec.~\ref{sec: RT of noncontextual wirings} the resource theory framed around noncontextual wirings introduced in Ref.~\cite{amaral2018noncontextual}. We suggest Sec.~\ref{sec: RT of noncontextual wirings} even to readers familiar with noncontextual wirings, as our presentation is significantly different from existing ones~\cite{amaral2018noncontextual,amaral2019resource,barbosa2023closing}. Our main results regarding the relation between NCW and LOSR are presented in Sec.~\ref{sec: NCW vs LOSR}. The convexity of type-dependent NCW operations is shown in Sec.~\ref{sec: convexity result}. In Sec.~\ref{sec: global} we present the comparability results described before and make some final remarks in Sec.~\ref{sec: discussion}.

\section{Background}\label{sec: background}

We begin by describing the basic mathematical elements of the resource theories framework~\cite{coecke2016mathematical,gonda2019monotones} in Sec.~\ref{sec: resource}. Although at first, the presentation will often seem quite abstract, we will complement the mathematical definitions with physically motivated interpretations. Sec.~\ref{sec: contextuality} introduces compatibility scenarios and the notion of Kochen-Specker noncontextuality. We then define the resource theory of Kochen-Specker contextuality where the free operations are noncontextual wirings in Sec.~\ref{sec: RT of noncontextual wirings}.

\subsection{Resource theories}
\label{sec: resource}

Let us start by defining an abstract resource theory:

\begin{definition}\label{def: resource theory}
    A \emph{resource theory} is a triplet $\mathcal{R} = (\mathcal{U},\mathcal{F},\mathcal{T})$ where $\mathcal{U}$ is a set of objects, $\mathcal{F} \subsetneq \mathcal{U}$ is a proper subset of objects called free, and $\mathcal{T}$ is a set of transformations $T: \mathcal{U} \to \mathcal{U}$ satisfying $T(\mathcal{F}) \subseteq \mathcal{F}$, called free transformations or free operations.
\end{definition}

Intuitively, $\mathcal{U}$ is the set of \emph{all objects} under study. For example, in static quantum resource theories~\cite{chitambar2019quantum}, $\mathcal{U} = \mathcal{D}(\mathcal{H})$, the set of all possible quantum states of a quantum system $\mathcal{H}$. Each object in this case is a quantum state $\rho \in \mathcal{D}(\mathcal{H})$. Above in Def.~\ref{def: resource theory}, the notion of sequential composition is given by the notion of composition of maps $T_2 \circ T_1$ as $T_2(T_1(\cdot))$. 

The subset $\mathcal{F}$ is the set of \emph{free} objects. The set of free objects corresponds to those that lack a certain property of interest, which we refer to as the \emph{resource}. Then, elements in $\mathcal{U} \setminus \mathcal{F}$ will correspond to \emph{resourceful} objects. We select some property of interest to be studied, but the terminology of \emph{resource} clearly suggests that we will \emph{use} resourceful objects to perform some task, that would not be possible without the presence of the resource.  

The definition of a resource theory has a clear operational intuition. However, being too strict with this intuition can sometimes lead to a very narrow view of the framework. For instance, objects in $\mathcal{F}$ are not necessarily those that are easily accessible in practice, with current technology, as the term \emph{free} suggests. For example, in the resource theory of nonstabilizerness~\cite{veitch2014resourcetheory,braviy2016improved}, crucial for quantum computation, highly entangled pure multipartite states are considered free. Yet, it is significantly challenging to prepare such states~\cite{bluvstein2023logical,moses2023race_track}. Neither the set $\mathcal{U}\setminus \mathcal{F}$ must be viewed as the set of scarce or rare objects. Almost every pure quantum state is entangled~\cite{gross2009most}, and still, this resource powers many information tasks~\cite{horodecki2009quantum}. Finally, objects in $\mathcal{U}\setminus \mathcal{F}$ should also not be considered as those that are hard to prepare with current technology. For instance, highly non-trivial coherent states can be prepared with current technology~\cite{giordani2021witnesses,giordani2023experimental}, and still, the resource theory of coherence is one of the most widely investigated resource theories in practice~\cite{baumgratz2014quantifying,streltsov17colloquium}.

Finally, the set $\mathcal{T}$ is considered the set of transformations that can be freely performed, in the precise sense that they do not generate resources. For entanglement theory, an example is the set of local operations and classical communications~\cite{horodecki2009quantum}. In this case, any free transformation $T$ acting on separable states will output another separable state. Another example of free operations in entanglement theory, that is a subset of the one just described, is the set of all local operations with shared randomness introduced in Ref.~\cite{dukaric2008limit} and subsequently developed in Refs.~\cite{buscemi2012all,deVicente2014LOSR,geller2014quantifying,wolfe2020quantifying,schmid2020typeindependent}.

When $\mathcal{U} \subseteq V$ is a subset of some vector space $V$ it is important to characterize the geometry of $\mathcal{F}$. For example, when the set of free objects is convex, some monotones such as the contextual fraction~\cite{abramsky2017contextual} and robustness~\cite{li2020contextualrobustness}, can be defined and calculated using standard linear programming techniques. Moreover, some general results within the resource theory approach, such as an operational interpretation using discrimination tasks~\cite{takagi2019general,takagi2019operational} and being able to map problems to conic programming problems~\cite{uola2019quantifying} are only possible if $\mathcal{F}$ is a convex (or also compact) set. It is also useful when the set $\mathcal{T}$ is convex, especially when $\mathcal{F}$ itself is convex, for consistency.  One of our main results (presented in  Sec.~\ref{sec: convexity result}) will be that the set of all noncontextual wirings of a fixed type is \emph{convex}. The notion of the `type' of a transformation will become clear later. 

\subsubsection{Pre-order induced by free operations}

We now introduce the idea of the \emph{pre-order} of objects in a resource theory. This idea intimately connects to (inter)conversion between objects via free operations. Given two objects $A, B \in \mathcal{U}$, we are interested to learn if there is a free transformation $T \in \mathcal{T}$ such that $T(A)=B$, denoted as $A \to B$. This problem has a practical motivation. If it is possible to freely obtain $B$ from $A$, it is sufficient to have access only to $A$ and one obtains $B$ (and a whole class of other objects) by freely acting on $A$. However, we will be interested in the somewhat foundational side of this question, where free transformations \emph{fix a structure} of which objects are comparable through the lenses of the operations in $\mathcal{T}$. 

Note that for most resource theories, there exists a free operation $A \to A$. One possibility is to assume that the identity map $\mathrm{id}: \mathcal{U} \to \mathcal{U}$ satisfying $\mathrm{id}(A)=A$ is a free operation. Moreover, free operations are also assumed to be transitive, i.e., if $A \to B$ and $B \to C$ then $A \to C$, for any $A, B, C \in \mathcal{U}$. The conjunction of these two properties implies that resource conversion $\to$ induces a \emph{pre-order} to the set $\mathcal{U}$. In mathematically oriented literature, one usually denotes the order relation as $\succeq$. In our case,
\begin{equation}
    A \succeq B\, \Leftrightarrow \exists \,T \in \mathcal{T}: T(A)=B.
\end{equation}

The notation $\succeq$ is more common in mathematical literature, while $\to$ is commonly used in physics-oriented literature, as it passes the intuitive idea that one object is being transformed into another. In what follows we will only make use of $\to$. It is, therefore, equivalent to defining a resource theory by specifying a set of free operations, or by specifying a certain pre-order of the objects $\mathcal{U}$. The fact that this relation is only a pre-order, as opposed to a total or a partial order,  introduces most of the `richness' of a given resource theory. This `richness' is represented by the plurality of possible (or impossible)  ways to compare different objects.

\begin{definition}[Equivalent, inequivalent and incomparable objects]\label{def: local comparability relations}
Let $\mathcal{R} = (\mathcal{U},\mathcal{F},\mathcal{T})$ be a resource theory and $B_1,B_2 \in \mathcal{U}$. We say that: 
\begin{itemize}
    \item $B_1$ is \emph{equivalent} to $B_2$ when $B_1 \rightarrow B_2$ and $B_2 \rightarrow B_1$. Otherwise, we say that $B_1$ and $B_2$ are \emph{inequivalent}.
    \item $B_1$ is \emph{incomparable} to $B_2$ when $B_1 \nrightarrow B_2$ and $B_2 \nrightarrow B_1$.
\end{itemize}
Here, we denote $A \nrightarrow B$ if there exists no $T \in \mathcal{T}$ such that $T(A)=B$. 
\end{definition}

From the perspective of the pre-order induced by free transformations, the definition above focuses on \emph{local comparability properties of a pair of objects}. Given two objects, how do they relate under free operations? Incomparability is particularly striking as, in a precise sense, discovering incomparable objects signals inequivalent classes of resourceful objects. In entanglement theory, for instance, it led to the discovery of different entanglement classes~\cite{dur2000three,verstraete2002four}.

\begin{lemma}[Free objects are equivalent]\label{Lemma: free equivalent}
    Consider $\mathcal{R} = (\mathcal{U},\mathcal{F},\mathcal{T})$ to be any resource theory. Suppose there exists $\star \in \mathcal{U}$ such that, for any $ A \in \mathcal{F}$ we have $\star \to A$ and  $ A \to \star$. Then, every pair $A, B \in \mathcal{F}$ is equivalent. 
\end{lemma}

\begin{proof}
   Due to $\star$, $B \to \star \to A$ and $A \to \star \to B$, for every $A,B \in \mathcal{F}$. But since $\to$ is a pre-order, it is transitive, and therefore $B \to A$ and $A \to B$, implying $A$ and $B$ are equivalent.
\end{proof}

From above, it is often the case that all free objects of a given resource theory are equivalent. One can interpret $\star$ as a reference object, and free objects as those that can be freely constructed from the reference. Moreover, one can imagine that any object can be discarded towards this reference object if we consider discarding to be a free operation. 

In this work, we also investigate \emph{global comparability properties}, which can be understood as comparability properties beyond a single pair. They provide a broader view of the action of the pre-order. We focus on the following global properties, already recognized as relevant elsewhere~\cite{wolfe2020quantifying,gonda2019monotones}:

\begin{definition}[Global properties]\label{def: global properties} 
Let $\mathcal{R} = (\mathcal{U},\mathcal{F},\mathcal{T})$ be a resource theory, and let $\to$ be the pre-order induced by the free operations $\mathcal{T}$. 
\begin{itemize}
    \item When the pre-order is such that every pair of objects is either strictly ordered or equivalent, the set of objects is said to be \emph{totally pre-ordered}.
    
    \item If for every triplet of objects $A, B, C \in \mathcal{U}$, all distinct, such that when $A \nleftrightarrow B$ and $B \nleftrightarrow C$ then $A \nleftrightarrow C$, we say that the pre-order is \emph{weak}.
    
    \item A \emph{chain} is a subset of objects where every pair of elements is strictly ordered. The \emph{height} of the pre-order is the cardinality of the largest chain in this pre-order.
    
    \item Likewise, an \emph{antichain} is a subset of elements where every pair of elements is incomparable. The \emph{width} of the pre-order is the cardinality of the largest antichain contained in the pre-order.
    
    \item We say that an object $B_2$ lies in the \emph{interval} of objects $B_1$ and $B_3$ iff both $B_1 \rightarrow B_2$ and $B_2 \rightarrow B_3$ hold. If the number of equivalence classes $[B]$ that lie in the interval of a pair of objects $(B_1, B_2)$ is finite for every pair of inequivalent objects $(B_1, B_2)$, we say that the pre-order is \emph{locally finite}, otherwise it is said to be \emph{locally infinite}.
\end{itemize}
\end{definition}

On the last point above, note that a pre-order does not define an equivalence relation, since it does not need to be symmetric. However, we can define a new relation $A \sim B:\iff A \to B$ and $B \to A$. In this case, this new relation \emph{is} an equivalence relation, and divides $\mathcal{U}$ into equivalence classes, that are those mentioned above in the last point of Def.~\ref{def: global properties}.

\subsubsection{Resource Monotones}\label{subsec: resource monotones}

One of the most important uses of a resource theory formalism is for quantifying the resources contained in objects of the theory. 

\begin{definition}[Resource monotones]
    Let $(\mathcal{U}, \mathcal{F}, \mathcal{T})$ be any resource theory. We define a \emph{resource monotone} as a pre-order preserving function $\mathsf{m}: \C{U} \r \Bar{\mathbb{R}}:=\mathbb{R} \cup \{-\infty,\infty \}$, such that for all $A, B \in \C{U}$,
\begin{equation*}
 A \to B \implies \mathsf{m}(B) \leq \mathsf{m}(A).
\end{equation*}
\end{definition}

Note that, if the conditions of Lemma~\ref{Lemma: free equivalent} apply the monotone is constant for all free resources. Since for every $A, B \in \mathcal{F}$ we have that $\mathsf{m}(A) \leq \mathsf{m}(B)$ \emph{and} $\mathsf{m}(A)\geq \mathsf{m}(B)$ we then have that $\forall A \in \mathcal{F}, \mathsf{m}(A) = \mathsf{m}_\star$ constant, which is usually taken to be equal to zero $\mathsf{m}_\star = 0$. 

Intuitively then, a monotone function gives a quantitative measure of how resourceful an object is. Because of their order-preserving property, these functions give us insightful information about the resource theory, as we will see in Sec.~\ref{sec: global}. If we write $B = T(A)$ (since $A \to B$) for some $T \in \mathcal{T}$ then $\mathsf{m}(B) = \mathsf{m}(T(A)) \leq \mathsf{m}(A)$. Or from another angle, for any two objects $A, B$ such that $\mathsf{m}(A)> \mathsf{m}(B)$ there exists no free transformation $T$ such that $T(B)=A$. Intuitively, this expresses the idea that since $A$ is more resourceful than $B$ one cannot freely go to $A$ from $B$.

It is worth mentioning that the pre-order structure of objects in a resource theory is more fundamental than any single resource monotone. Clearly, any monotone maps the pre-order over $\mathcal{U}$ into the total order of real numbers. A resource monotone captures certain aspects of the pre-order by assigning numerical values to the objects, but unless the pre-order is a total order (implying that all elements in $\mathcal{U}$ are comparable), it can never contain the complete information available in the pre-order \cite{amaral2019resource}. Even though there were early works (specifically in early advances in entanglement theory) in which one of the goals was to find what would be \ti{the} correct resource monotone, once incomparable objects were discovered, it became clear that this was a wrong path. The pre-order is the fundamental structure, with any particular resource monotone being a coarse-grained description of the resource theory. 

It is unclear if resource monotones can be used to investigate \emph{global} properties, given that they totally order objects from $\mathcal{U}$. Surprisingly, this is the case. As it was highlighted in  Ref.~\cite{gonda2019monotones}, global properties of the pre-order can be characterized by finding \emph{sufficiently many} resource monotones~\cite{wolfe2020quantifying}. Another example is Ref.~\cite{duarte2018concentration}, which introduced contextuality monotones to study geometrical aspects of particular sets of possible behaviors (free objects or not) inside and outside the quantum set of correlations. Later in our work, we will use two resource monotones (so-called cost and yield monotones) to investigate the global comparability properties of the pre-order induced by noncontextual wirings.

\subsection{Kochen--Specker Contextuality}
\label{sec: contextuality}

Contextuality can be viewed as the impossibility of thinking about statistical results of measurements as revealing pre-existing objective properties of that system,  which are independent of the actual set of measurements one chooses to make~\cite{budroni2021quantum, kochen1975problem}. In our treatment, it is a property of behaviors defined with respect to compatibility scenarios (a.~k.~a. measurement scenarios~\cite{barbosa2023closing}).

\subsubsection{Compatibility scenarios}

\begin{definition}
    [From Refs.~\cite{amaral2018noncontextual,amaral2019resource, amaral2018graph}] A \emph{compatibility scenario} is a triplet $\Upsilon := (\C{M},\C{C},\C{O}^{\mathcal{M}})$, where $\C{M}$ is a finite set of measurements,  $\C{C}$ is a family of subsets of $\C{M}$, called maximal contexts, and $\mathcal{O}^{\mathcal{M}} = \prod_{x \in \mathcal{M}}\mathcal{O}^x$, where $\mathcal{O}^x$ are the outcomes of $x \in \mathcal{M}$. For all $\gamma,\gamma' \in \C{C}$,   $\gamma \subseteq \gamma'$ implies $\gamma = \gamma'$.
\end{definition}

Each context $\gamma \in \C{C}$ represents a set of measurements in $\C{M}$ that can be jointly performed. For each context $\gamma$, the set of all possible outcomes for the joint measurement of the measurements in $\gamma$ is the set $\C{O}^\gamma$. When we jointly perform the measurements of $\gamma$, our output is encoded in a tuple $\tb{s} \in \C{O}^\gamma$. Later we will consider scenarios where $\mathcal{O}^x = \mathcal{O}^{x'}$ for all $x,x' \in \mathcal{M}$, i.e. where all measurements have the same outcomes, and we will simply denote such sets of outcomes as $\mathcal{O}^x = \mathcal{O}^{x'} = \mathcal{O}$. Whenever this happens, we will also simply write $\Upsilon = (\mathcal{M},\mathcal{C},\mathcal{O})$ instead of $(\mathcal{M},\mathcal{C},\mathcal{O}^{\mathcal{M}})$ to simplify the notation.

\subsubsection{Behaviors (or boxes)}

Kochen-Specker noncontextuality will be viewed as a constraint satisfied by probabilistic data over compatibility scenarios, that we now define:

\begin{definition}[Behaviors]
    Given a scenario $\Upsilon =(\C{M},\C{C},\C{O}^{\C{M}})$, a \emph{behavior} (a. k. a. a box) $B$ in this scenario is a family of probability distributions, one for each maximal context $\gamma \in \C{C}$,

\begin{equation}\label{eq: behavior}
    B = \Bigg\{ p_\gamma : \C{O}^\gamma \r [0,1]  \Bigg |  \sum_{\tb{s} \in \C{O}^\gamma} p_\gamma (\tb{s}) = 1, \gamma \in \C{C} \Bigg\}.
\end{equation}
\end{definition}

Experimentally, behaviors are the result of running many times a protocol that prepares a certain system and performs sequential (ideal) measurements in $\gamma \in \mathcal{C}$ returning joint outcomes $\mathbf{s} \in \mathcal{O}^\gamma$. Because of that, they are also called correlations, probabilistic data-tables, or simply data-tables. They are only well-defined with respect to a compatibility scenario. 

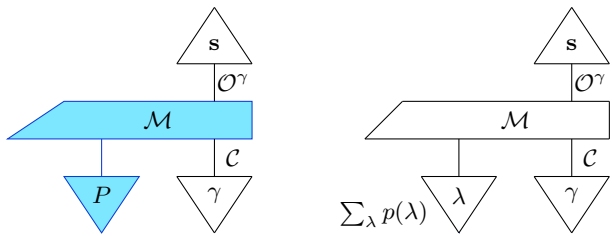
\begin{figure}
\centering
\begin{tikzpicture}
	\begin{pgfonlayer}{nodelayer}
		\node [style=none] (0) at (-2, -0.25) {};
		\node [style=none] (1) at (0.5, -0.25) {};
		\node [style=none] (2) at (0.5, -0.75) {};
		\node [style=none] (3) at (-2.75, -0.75) {};
		\node [style=none] (4) at (-2, -1.25) {};
		\node [style=none] (5) at (-1, -1.25) {};
		\node [style=none] (6) at (-1.5, -2) {};
		\node [style=none] (7) at (-1.5, -2) {};
		\node [style=none] (8) at (-1.5, -1.25) {};
		\node [style=none] (9) at (-1.5, -0.75) {};
		\node [style=none] (16) at (0, -0.25) {};
		\node [style=none] (17) at (0, -0.75) {};
		\node [style=none] (18) at (0, 0.25) {};
		\node [style=none] (19) at (0, -1.25) {};
		\node [style=none] (22) at (-0.5, 0.25) {};
		\node [style=none] (23) at (0.5, 0.25) {};
		\node [style=none] (24) at (0, 1) {};
		\node [style=none] (25) at (0, 0.5) {$\mathbf{s}$};
		\node [style=none] (26) at (-1.5, -1.5) {$P$};
		\node [style=none] (27) at (-0.75, -0.5) {$\mathcal{M}$};
		\node [style=none] (28) at (-0.5, -1.25) {};
		\node [style=none] (29) at (0.5, -1.25) {};
		\node [style=none] (30) at (0, -2) {};
		\node [style=none] (31) at (0, -1.5) {$\gamma$};
		\node [style=none] (33) at (0.25, 0) {$\mathcal{O}^\gamma$};
		\node [style=none] (34) at (0.25, -1) {$\mathcal{C}$};
		\node [style=none] (36) at (5.25, -0.25) {};
		\node [style=none] (37) at (5.25, -0.75) {};
		\node [style=none] (38) at (2, -0.75) {};
		\node [style=none] (39) at (2.75, -1.25) {};
		\node [style=none] (41) at (3.25, -2) {};
		\node [style=none] (42) at (3.25, -2) {};
		\node [style=none] (43) at (3.25, -1.25) {};
		\node [style=none] (44) at (3.25, -0.75) {};
		\node [style=none] (45) at (4.75, -0.25) {};
		\node [style=none] (46) at (4.75, -0.75) {};
		\node [style=none] (47) at (4.75, 0.25) {};
		\node [style=none] (48) at (4.75, -1.25) {};
		\node [style=none] (49) at (4.25, 0.25) {};
		\node [style=none] (50) at (5.25, 0.25) {};
		\node [style=none] (51) at (4.75, 1) {};
		\node [style=none] (52) at (4.75, 0.5) {$\mathbf{s}$};
		\node [style=none] (53) at (3.25, -1.5) {$\lambda$};
		\node [style=none] (54) at (4, -0.5) {$\mathcal{M}$};
		\node [style=none] (55) at (4.25, -1.25) {};
		\node [style=none] (56) at (5.25, -1.25) {};
		\node [style=none] (57) at (4.75, -2) {};
		\node [style=none] (58) at (4.75, -1.5) {$\gamma$};
		\node [style=none] (59) at (5, 0) {$\mathcal{O}^\gamma$};
		\node [style=none] (60) at (5, -1) {$\mathcal{C}$};
		\node [style=none] (61) at (2.5, -0.25) {};
		\node [style=none] (62) at (3.75, -1.25) {};
		\node [style=none] (63) at (2.25, -1.75) {$\sum_\lambda p(\lambda)$};
	\end{pgfonlayer}
	\begin{pgfonlayer}{edgelayer}
		\draw [style=alice] (1.center)
			 to (0.center)
			 to (3.center)
			 to (2.center)
			 to cycle;
		\draw [style=alice] (5.center)
			 to (4.center)
			 to (7.center)
			 to (6.center)
			 to cycle;
		\draw [style=alice] (8.center) to (9.center);
		\draw (16.center) to (18.center);
		\draw (17.center) to (19.center);
		\draw (23.center) to (22.center);
		\draw (22.center) to (24.center);
		\draw (24.center) to (23.center);
		\draw (28.center) to (29.center);
		\draw (28.center) to (30.center);
		\draw (30.center) to (29.center);
		\draw [style=alice] (42.center) to (41.center);
		\draw (45.center) to (47.center);
		\draw (46.center) to (48.center);
		\draw (50.center) to (49.center);
		\draw (49.center) to (51.center);
		\draw (51.center) to (50.center);
		\draw (55.center) to (56.center);
		\draw (55.center) to (57.center);
		\draw (57.center) to (56.center);
		\draw (38.center) to (61.center);
		\draw (61.center) to (36.center);
		\draw (36.center) to (37.center);
		\draw (37.center) to (38.center);
		\draw (39.center) to (62.center);
		\draw (62.center) to (42.center);
		\draw (42.center) to (39.center);
		\draw (43.center) to (44.center);
	\end{pgfonlayer}
\end{tikzpicture}

\caption{Diagrammatic representation of  $p_\gamma(\mathbf{s})$ in a compatibility scenario. (Left) Given a certain classical input $\gamma \in \mathcal{C}$ and a prepared system $P$ we perform joint measurements (or equivalently, sequential ideal measurements) selected from a set  $\mathcal{M}$ and obtain joint outcomes $\mathbf{s} \in \mathcal{O}^\gamma$. Labels in the wires denote the type of classical information they carry. Blue (colored) regions represent generic operational primitives (quantum, classical, or post-quantum). White (non-colored) regions represent classical operational primitives only. (Right) When we have a noncontextual behavior, each element $p_\gamma (\mathbf{s})$ in it takes the form of Eq.~\eqref{eq: noncontextual factorizable}. We can interpret the preparation $P$  as some state $\lambda$ sampled according to some randomness source $p(\lambda)$. In this case, the white (non-colored) regions depict classicality (noncontextuality). \label{fig: behavior} }
\end{figure}

Sometimes we also call a behavior a \emph{box} (see Fig.~\ref{fig: behavior}). Both terms are encountered in the literature, and we will use them interchangeably in our work. The intuition behind the term box is the following: Imagine the elements of $\C{M}$ as buttons of the box, and, for each measurement $x$, we imagine the box having $|\C{O}^x|$ output lights that inform us of the result of the measurements. The box has, therefore, certain rules as certain buttons cannot be jointly pressed (corresponding to certain measurements being incompatible). The information of allowed buttons to be jointly pressed is provided by maximal contexts $\gamma$. In this view, instead of sequential measurements one, equivalently, imagines that an experimental implementation is performing jointly all compatible measurements.

Behaviors may or may not satisfy what we call the \ti{no-disturbance condition}. Given two contexts
$\gamma$ and $\gamma '$, no-disturbance implies that the marginals for their intersection are well defined, and agree. If we have, for example, $\gamma = \{x, y\}$ and $\gamma' = \{y, z\}$, the no-disturbance condition implies:

\begin{equation*}
    \sum_{a} p_{\{x,y\}} (a,b) = \sum_{c} p_{\{y,z\}} (b,c).
\end{equation*}
\begin{definition}[No-disturbance set of behaviors]
 The \emph{no-disturbance set} $\mathrm{ND}(\Upsilon)$ is the set of behaviors that satisfy the no-disturbance condition for any intersection of contexts in the scenario $\Upsilon$.
\end{definition}

The defining idea of noncontextuality is the possibility of assigning a single probability distribution to the whole set $\C{O}^\C{M}$, that has marginals in each maximal context consistent with the behavior $B$. We call this probability distribution $p_\C{M} : \C{O}^\C{M} \r [0,1]$ a \ti{global section} for the scenario, that satisfy
\begin{equation}
    p_{\mathcal{M}}\vert_{\gamma}(\mathbf{s}) := \sum_{\mathbf{t} \in \mathcal{O}^\mathcal{M}:\mathbf{t}|_{\gamma}=\mathbf{s}}p_{\mathcal{M}}(\mathbf{t}) = p_\gamma(\mathbf{s}) 
\end{equation}
for all contexts $\gamma \in \mathcal{C}$ of $\Upsilon$ and all $\mathbf{s}\in \mathcal{O}^\gamma$. If this is possible, we say that the behavior $B=\{\{p_\gamma(\mathbf{s})\}_{\mathbf{s}\in \mathcal{O}^\gamma}\}_{\gamma \in \mathcal{C}}$ is KS-noncontextual.

\begin{definition}[Noncontextual set of behaviors]
 The \emph{noncontextual} set $\mathrm{NC}(\Upsilon)$ is the set of all behaviors for which there exists a global section with marginals over the maximal contexts of $\Upsilon$ returning the same distributions of the behavior.
\end{definition}

From the Abramsky-Brandenburger Theorem~\cite{abramsky2011sheaf}, KS-noncontextual behaviors can be equivalently written as

\begin{equation}\label{eq: noncontextual factorizable}
    p_\gamma (\tb{s}) = \sum_\lambda p(\lambda) \prod_{\gamma_i \in \gamma} p_{\gamma_i} (s_i|\lambda),
\end{equation}
where $\lambda \in \Lambda$ are any set of variables, and $p(\lambda)$ is a probability distribution over these variables, i.e., satisfies $\sum_\lambda p(\lambda)=1$ and $0 \leq p(\lambda) \leq 1, \forall \lambda$. Also, $p_{\gamma_i}(s_i|\lambda)$ are so-called response functions, satisfying that for any given $\lambda$ the mapping $p_{\gamma_i}(\cdot|\lambda)$ yields a valid probability distribution over $\mathcal{O}^{\gamma_i}$, for any $\gamma_i \in \gamma$ and also any $\gamma \in \mathcal{C}$. The above description has been historically linked to the existence of a noncontextual hidden-variable model for the behavior. When a behavior $B \equiv \{\{p_\gamma(\mathbf{s})\}_{\mathbf{s}}\}_\gamma$ has this precise form, it is said to be factorizable~\cite{abramsky2011sheaf,barbosa2022continuous}. All noncontextual behaviors satisfy the no-disturbance condition. Different notions of noncontextuality have been proposed for behaviors that do not respect no-disturbance~\cite{kujala2015necessary,amaral2018necessary,amaral2019characterizing,tezzin2020contextualitybydefault}. Any such approach will have certain drawbacks~\cite{tezzin2022impossibility}.

We say that $\mathcal{R} = (\mathcal{U}_{\mathrm{ND}},\mathcal{F}_{\mathrm{NC}},\mathcal{T}_{\mathrm{NC}})$ is a resource theory of KS-contextuality if $\mathcal{U}_{\mathrm{ND}}:= \bigsqcup_{\Upsilon}\mathbf{\mathrm{ND}}(\Upsilon)$ and $\mathcal{F}_{\mathrm{NC}}:=\bigsqcup_{\Upsilon}\mathbf{\mathrm{NC}}(\Upsilon)$, and the (disjoint) union is taken over all possible compatibility scenarios. For each fixed scenario $\Upsilon$ the objects are nondisturbing behaviors and free objects are noncontextual behaviors.  Free operations are those that preserve the noncontextuality of behaviors mapped across different compatibility scenarios. In this work, we focus on a specific set of free operations, that we will characterize later. Other operations can be taken as free, as was done in Refs.~\cite{abramsky2019comonadic,barbosa2023closing,karvonen2021neither}, that considered as free a broader class of operations than the noncontextual wirings. 

\subsubsection{The $n$-cycle noncontextuality inequalities.}

The only infinite family of compatibility scenarios that have been completely characterized is the family of $n$-cycle scenarios~\cite{araujo2013all}, which we will denote as $\Upsilon_n$. In such scenarios, one has $n$ dichotomic measurements  $\mathcal{M} = \{x_i\}_{i=1}^n$ and all maximal contexts are given by $\{x_i,x_{i+1}\}$, where here summation is taken to be module $n$. All vertices and facet-defining inequalities of the polytope $\mathrm{ND}(\Upsilon_n)$ are known~\cite{araujo2013all}. Moreover, the set of all facet-defining inequalities of $\mathrm{NC}(\Upsilon_n)$ is also known, and given by 

\begin{equation}\label{eq: noncontextuality inequalities}
    I_k^{(n)}(B) = \sum^{n-1}_{i=0} a_i \langle x_i x_{i+1} \rangle \leq n-2,
\end{equation}
with each value of $k$ being associated with a particular choice of values for $a_i \in \{-1, -1\}$ such that the number of terms $a_i = -1$ is odd. Above, we are mapping behaviors $B$ from $\Upsilon_n$ to two-point correlation functions via, letting $p_{\{x_i,x_j\}}(ab) \equiv p(ab|x_ix_j)$, 
\begin{align*}
\langle x_i x_{i+1}\rangle &= +p(00|x_ix_{i+1})+p(11|x_ix_{i+1})\\&-p(10|x_ix_{i+1})-p(01|x_ix_{i+1}).
\end{align*}Each label $k$ is therefore associated with a different facet of $\mathrm{NC}(\Upsilon_n)$, for each fixed choice $n$. The sets of points $\{\langle x_i x_j\rangle \}_{i,j}$ are called the \emph{sets of correlations}. In Sec.~\ref{sec: global} we will investigate how noncontextual wirings order such sets, focusing on the scenarios $\Upsilon_n$. For such scenarios, there is a one-to-one correspondence between the set of behaviors and the set of correlations~\cite{araujo2013all}.

Various properties of the polytopes $\mathrm{NC}(\Upsilon_n)$ are known. For every contextual behavior $B$ there is a unique $k$ for which $I_k^{(n)} (B) > n-2$. Ref.~\cite{choudhary2024lifting} investigated liftings of these inequalities to other compatibility scenarios. The values for the maximal quantum violations of the inequalities~\eqref{eq: noncontextuality inequalities} are known, and given by~\cite{araujo2013all}:

\be
    I_{Q}^{\mathrm{max}} = \begin{cases} \frac{3n \cos{(\frac{\pi}{n})} -n}{1 + \cos{(\frac{\pi}{n})}}, & \mbox{for odd } n, \\ n \cos{(\frac{\pi}{n})}, & \mbox{for even } n. \end{cases}
\ee

Behaviors for which the value of the $I$ function is larger than $I_{Q}^{\mathrm{max}}$ will also be of interest. Specifically, those that reach the algebraic maximum $I_k^{(n)}(B) = n$ will be used later. 

\subsection{Resource theory of Kochen-Specker contextuality under wirings}\label{sec: RT of noncontextual wirings}

\subsubsection{Pre- and post-processing}

To define the free operations of our resource theory, we begin by defining certain special operations that take behaviors (our objects) in a given scenario into other behaviors, potentially in another scenario.

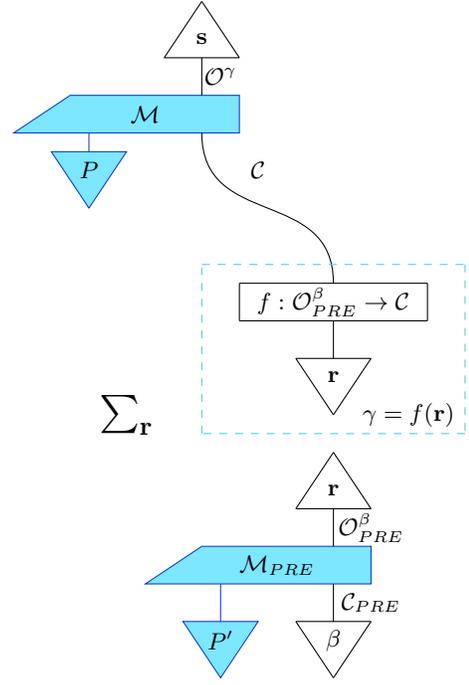
\begin{figure}[h]
    \centering
    \begin{tikzpicture}
	\begin{pgfonlayer}{nodelayer}
		\node [style=none] (0) at (-1.75, 0.25) {};
		\node [style=none] (1) at (0.5, 0.25) {};
		\node [style=none] (2) at (0.5, -0.25) {};
		\node [style=none] (3) at (-2.5, -0.25) {};
		\node [style=none] (4) at (-2, -0.5) {};
		\node [style=none] (5) at (-1, -0.5) {};
		\node [style=none] (6) at (-1.5, -1.25) {};
		\node [style=none] (7) at (-1.5, -1.25) {};
		\node [style=none] (8) at (-1.5, -0.5) {};
		\node [style=none] (9) at (-1.5, -0.25) {};
		\node [style=none] (16) at (0, 0.25) {};
		\node [style=none] (17) at (0, -0.25) {};
		\node [style=none] (18) at (0, 0.75) {};
		\node [style=none] (22) at (-0.5, 0.75) {};
		\node [style=none] (23) at (0.5, 0.75) {};
		\node [style=none] (24) at (0, 1.5) {};
		\node [style=none] (25) at (0, 1) {$\mathbf{s}$};
		\node [style=none] (26) at (-1.5, -0.75) {$P$};
		\node [style=none] (27) at (-0.75, 0) {$\mathcal{M}$};
		\node [style=none] (33) at (0.25, 0.5) {$\mathcal{O}^\gamma$};
		\node [style=none] (34) at (0.75, -0.75) {$\mathcal{C}$};
		\node [style=none] (35) at (0, -5.75) {};
		\node [style=none] (36) at (2.25, -5.75) {};
		\node [style=none] (37) at (2.25, -6.25) {};
		\node [style=none] (38) at (-0.75, -6.25) {};
		\node [style=none] (39) at (-0.25, -6.75) {};
		\node [style=none] (40) at (0.75, -6.75) {};
		\node [style=none] (41) at (0.25, -7.5) {};
		\node [style=none] (42) at (0.25, -7.5) {};
		\node [style=none] (43) at (0.25, -6.75) {};
		\node [style=none] (44) at (0.25, -6.25) {};
		\node [style=none] (45) at (1.75, -5.75) {};
		\node [style=none] (46) at (1.75, -6.25) {};
		\node [style=none] (47) at (1.75, -5.25) {};
		\node [style=none] (48) at (1.75, -6.75) {};
		\node [style=none] (49) at (1.25, -5.25) {};
		\node [style=none] (50) at (2.25, -5.25) {};
		\node [style=none] (51) at (1.75, -4.5) {};
		\node [style=none] (52) at (1.75, -5) {$\mathbf{r}$};
		\node [style=none] (53) at (0.25, -7) {$P'$};
		\node [style=none] (54) at (1, -6) {$\mathcal{M}_{PRE}$};
		\node [style=none] (55) at (1.25, -6.75) {};
		\node [style=none] (56) at (2.25, -6.75) {};
		\node [style=none] (57) at (1.75, -7.5) {};
		\node [style=none] (58) at (1.75, -7) {$\beta$};
		\node [style=none] (59) at (2.25, -5.5) {$\mathcal{O}^\beta_{PRE}$};
		\node [style=none] (60) at (2.25, -6.5) {$\mathcal{C}_{PRE}$};
		\node [style=none] (61) at (1.75, -3.25) {};
		\node [style=none] (62) at (2.25, -3.25) {};
		\node [style=none] (63) at (1.25, -3.25) {};
		\node [style=none] (64) at (1.75, -4) {};
		\node [style=none] (65) at (1.75, -3.5) {$\mathbf{r}$};
		\node [style=none] (66) at (0.5, -2.25) {};
		\node [style=none] (67) at (3, -2.25) {};
		\node [style=none] (68) at (0.5, -2.75) {};
		\node [style=none] (69) at (3, -2.75) {};
		\node [style=none] (70) at (1.75, -2.75) {};
		\node [style=none] (71) at (1.75, -2.25) {};
		\node [style=none] (72) at (1.75, -2.5) {$f:\mathcal{O}_{PRE}^\beta \to \mathcal{C}$ };
		\node [style=none] (73) at (0, -2) {};
		\node [style=none] (74) at (3.5, -2) {};
		\node [style=none] (75) at (0, -4.25) {};
		\node [style=none] (76) at (3.5, -4.25) {};
		\node [style=none] (77) at (2.75, -4) {$\gamma=f(\mathbf{r})$};
		\node [style=none] (78) at (-1, -4) {\Large$\sum_{\mathbf{r}}$};
	\end{pgfonlayer}
	\begin{pgfonlayer}{edgelayer}
		\draw [style=alice] (1.center)
			 to (0.center)
			 to (3.center)
			 to (2.center)
			 to cycle;
		\draw [style=alice] (5.center)
			 to (4.center)
			 to (7.center)
			 to (6.center)
			 to cycle;
		\draw [style=alice] (8.center) to (9.center);
		\draw (16.center) to (18.center);
		\draw (23.center) to (22.center);
		\draw (22.center) to (24.center);
		\draw (24.center) to (23.center);
		\draw [style=alice] (36.center)
			 to (35.center)
			 to (38.center)
			 to (37.center)
			 to cycle;
		\draw [style=alice] (40.center)
			 to (39.center)
			 to (42.center)
			 to (41.center)
			 to cycle;
		\draw [style=alice] (43.center) to (44.center);
		\draw (45.center) to (47.center);
		\draw (46.center) to (48.center);
		\draw (50.center) to (49.center);
		\draw (49.center) to (51.center);
		\draw (51.center) to (50.center);
		\draw (55.center) to (56.center);
		\draw (55.center) to (57.center);
		\draw (57.center) to (56.center);
		\draw (63.center) to (62.center);
		\draw (62.center) to (64.center);
		\draw (64.center) to (63.center);
		\draw (68.center) to (69.center);
		\draw (69.center) to (67.center);
		\draw (67.center) to (66.center);
		\draw (66.center) to (68.center);
		\draw [in=-90, out=90, looseness=0.75] (61.center) to (70.center);
		\draw [in=-90, out=90, looseness=1.25] (71.center) to (17.center);
		\draw [style={alice_dashed}] (75.center) to (73.center);
		\draw [style={alice_dashed}] (73.center) to (74.center);
		\draw [style={alice_dashed}] (74.center) to (76.center);
		\draw [style={alice_dashed}] (76.center) to (75.center);
	\end{pgfonlayer}
\end{tikzpicture}

    \caption{Diagrammatic representation of a pre-processing defined by $p_\beta(\mathbf{s})$ from Eq.~\eqref{eq: pre-processing behavior}. We start with some initial behavior $B_{PRE}$ (lower box) that outputs some set of outcomes. Using these outcomes, and a function $f$ that consistently maps outcomes $\mathbf{r} \in \mathcal{O}^\beta_{PRE}$ of this behavior to contexts $\mathcal{C}$ in another scenario, we input some choice of context $\gamma = f(\mathbf{r}) \in \mathcal{C}$ to perform measurements in $\mathcal{M}$. Note that we average over all possible outcomes $\mathbf{r}$, to get the probability $p_{\beta}(\mathbf{s})$. \label{fig: pre-processing boxes}}
\end{figure}

One of the basic operations is the operation of \ti{pre-processing} a behavior. Assume that we start with a given scenario $\Upsilon = (\mathcal{M},\mathcal{C},\mathcal{O}^{\mathcal{M}})$. We introduce a new scenario $\Upsilon_{PRE} = (\C{M}_{PRE}, \C{C}_{PRE}, \C{O}_{PRE})$\footnote{To ease the notation, we write here and in the remaining of our work $\mathcal{O}_{PRE} \equiv \mathcal{O}_{PRE}^{\mathcal{M}_{PRE}}$, and similarly for the scenarios $\Upsilon_{POS}$.}, with new measurements, contexts and measurement outcomes. Behaviors in this scenario are denoted as  $B_{PRE}$. We associate each output of $B_{PRE}$ with an input of $B$, in such a way that every output configuration of $B_{PRE}$ defines a possible input configuration in $B$, that is, associated with every output $\tb{r} \in \C{O}_{PRE}$, we have a possible context $f(\tb{r}) = \gamma \in \C{C}$.

With this, we define a new behavior $\C{W}_{PRE}(B)$ given by

\begin{equation}\label{eq: pre-processing behavior}
    p_\beta (\tb{s}) = \sum_{\tb{r} \in \mathcal{O}_{PRE}^\beta} p_\beta (\tb{r}) p_{\gamma=f(\tb{r})} (\tb{s}),
\end{equation}
where the sum runs over all outputs $\tb{r}$ associated with the context $\beta$ in $\mathcal{C}_{PRE}$. The probability distribution $p_\beta(\mathbf{r})$ is given by the behavior used for the pre-processing. Figure~\ref{fig: pre-processing boxes} presents a diagrammatic representation of such a behavior. The function $f: \mathcal{O}^\beta \to \mathcal{C}$ translates information of the outputs of a behavior $B_{PRE}$ in information for the inputs of a behavior in $B$, that corresponds to information about the context. Moreover, this function $f$ can be any, insofar as its image is a subset of $\mathcal{C}$. 

Analogously, we can define the \ti{post-processing} of a behavior. We again introduce $\Upsilon_{POS} = (\C{M}_{POS}, \C{C}_{POS}, \C{O}_{POS})$ together with a behavior $B_{POS}$. The same association is made between outputs $\tb{s} \in \C{O}^{\gamma}$ and contexts $g(\tb{s}) = \delta \in \C{M}_{POS}$. The new behavior obtained $\C{W}_{POS} (B)$ is given by

\begin{equation}\label{eq: post-processing behavior}
    p_\gamma (\tb{t}) = \sum_{\tb{s} \in \mathcal{O}^\gamma} p_\gamma (\tb{s}) p_{\delta=g (\tb{s})} (\tb{t}).
\end{equation}

We provide a diagrammatic representation of this behavior in Fig.~\ref{fig: post-processing boxes}. The probability $p_\gamma(\mathbf{s})$ is the probability that upon inputting the classical information of the context $\gamma$ we obtain the joint outcome $\mathbf{s}$. This joint outcome is then consistently mapped, by some function $g:\mathcal{O}^\gamma \to \mathcal{C}_{POS}$ towards some choice of context $\delta=g(\mathbf{s}) \in \mathcal{C}_{POS}$. At this point, $g$ can be any function. Later, when we allow  $g$ to have information about inputs and outputs of $B_{PRE}$, this same function will need to be restricted, for the operation to remain free.

Both pre- and post-processing behaviors introduced were considered in full generality. This operational description is valid for KS-noncontextual, quantum, and post-quantum behaviors. The free operations in a resource theory of KS-contextuality under wirings will be those that are adaptive combinations of pre- and post-processing where the behaviors used to construct the pre- and post-processings are KS-noncontextual. 

\begin{figure}
    \centering
    \begin{tikzpicture}
	\begin{pgfonlayer}{nodelayer}
		\node [style=none] (0) at (0.75, -0.75) {};
		\node [style=none] (1) at (2.75, -0.75) {};
		\node [style=none] (2) at (2.75, -1.25) {};
		\node [style=none] (3) at (0, -1.25) {};
		\node [style=none] (4) at (0.25, -1.75) {};
		\node [style=none] (5) at (1.25, -1.75) {};
		\node [style=none] (6) at (0.75, -2.5) {};
		\node [style=none] (7) at (0.75, -2.5) {};
		\node [style=none] (8) at (0.75, -1.75) {};
		\node [style=none] (9) at (0.75, -1.25) {};
		\node [style=none] (16) at (2.25, -0.75) {};
		\node [style=none] (17) at (2.25, -1.25) {};
		\node [style=none] (18) at (2.25, -0.25) {};
		\node [style=none] (22) at (1.75, -0.25) {};
		\node [style=none] (23) at (2.75, -0.25) {};
		\node [style=none] (24) at (2.25, 0.5) {};
		\node [style=none] (25) at (2.25, 0) {$\mathbf{s}$};
		\node [style=none] (26) at (0.75, -2) {$P$};
		\node [style=none] (27) at (1.5, -1) {$\mathcal{M}$};
		\node [style=none] (33) at (2.5, -0.5) {$\mathcal{O}^\gamma$};
		\node [style=none] (34) at (2.5, -1.5) {$\mathcal{C}$};
		\node [style=none] (35) at (2.75, 5.25) {};
		\node [style=none] (36) at (5, 5.25) {};
		\node [style=none] (37) at (5, 4.75) {};
		\node [style=none] (38) at (2, 4.75) {};
		\node [style=none] (39) at (2.5, 4.5) {};
		\node [style=none] (40) at (3.5, 4.5) {};
		\node [style=none] (41) at (3, 3.75) {};
		\node [style=none] (42) at (3, 3.75) {};
		\node [style=none] (43) at (3, 4.5) {};
		\node [style=none] (44) at (3, 4.75) {};
		\node [style=none] (45) at (4.5, 5.25) {};
		\node [style=none] (46) at (4.5, 4.75) {};
		\node [style=none] (47) at (4.5, 5.75) {};
		\node [style=none] (48) at (2.25, 2.75) {};
		\node [style=none] (49) at (4, 5.75) {};
		\node [style=none] (50) at (5, 5.75) {};
		\node [style=none] (51) at (4.5, 6.5) {};
		\node [style=none] (52) at (4.5, 6) {$\mathbf{t}$};
		\node [style=none] (53) at (3, 4.25) {$P''$};
		\node [style=none] (54) at (3.75, 5) {$\mathcal{M}_{POS}$};
		\node [style=none] (59) at (5, 5.5) {$\mathcal{O}^\delta_{POS}$};
		\node [style=none] (60) at (4.5, 3.75) {$\mathcal{C}_{POS}$};
		\node [style=none] (61) at (2.25, -1.75) {};
		\node [style=none] (62) at (2.75, -1.75) {};
		\node [style=none] (63) at (1.75, -1.75) {};
		\node [style=none] (64) at (2.25, -2.5) {};
		\node [style=none] (65) at (2.25, -2) {$\gamma$};
		\node [style=none] (66) at (2.25, 1.75) {};
		\node [style=none] (67) at (1.75, 1.75) {};
		\node [style=none] (68) at (2.75, 1.75) {};
		\node [style=none] (69) at (2.25, 1) {};
		\node [style=none] (70) at (2.25, 1.5) {$\mathbf{s}$};
		\node [style=none] (71) at (0.75, 2.75) {};
		\node [style=none] (72) at (3.25, 2.75) {};
		\node [style=none] (73) at (0.75, 2.25) {};
		\node [style=none] (74) at (3.25, 2.25) {};
		\node [style=none] (75) at (2, 2.5) {$g: \mathcal{O}^\gamma \to \mathcal{C}_{POS}$};
		\node [style=none] (76) at (2.25, 2.25) {};
		\node [style=none] (77) at (0.25, 3.25) {};
		\node [style=none] (78) at (4, 3.25) {};
		\node [style=none] (79) at (0.25, 0.75) {};
		\node [style=none] (80) at (4, 0.75) {};
		\node [style=none] (81) at (3.25, 1.25) {$\delta=g(\mathbf{s})$};
		\node [style=none] (82) at (-0.5, 1.25) {\Large$\sum_{\mathbf{s}}$};
	\end{pgfonlayer}
	\begin{pgfonlayer}{edgelayer}
		\draw [style=alice] (1.center)
			 to (0.center)
			 to (3.center)
			 to (2.center)
			 to cycle;
		\draw [style=alice] (5.center)
			 to (4.center)
			 to (7.center)
			 to (6.center)
			 to cycle;
		\draw [style=alice] (8.center) to (9.center);
		\draw (16.center) to (18.center);
		\draw (23.center) to (22.center);
		\draw (22.center) to (24.center);
		\draw (24.center) to (23.center);
		\draw [style=alice] (36.center)
			 to (35.center)
			 to (38.center)
			 to (37.center)
			 to cycle;
		\draw [style=alice] (40.center)
			 to (39.center)
			 to (42.center)
			 to (41.center)
			 to cycle;
		\draw [style=alice] (43.center) to (44.center);
		\draw (45.center) to (47.center);
		\draw [in=90, out=-90] (46.center) to (48.center);
		\draw (50.center) to (49.center);
		\draw (49.center) to (51.center);
		\draw (51.center) to (50.center);
		\draw (63.center) to (62.center);
		\draw (62.center) to (64.center);
		\draw (64.center) to (63.center);
		\draw (61.center) to (17.center);
		\draw (67.center) to (68.center);
		\draw (67.center) to (69.center);
		\draw (69.center) to (68.center);
		\draw (73.center) to (71.center);
		\draw (71.center) to (72.center);
		\draw (72.center) to (74.center);
		\draw (74.center) to (73.center);
		\draw [in=90, out=-90, looseness=0.75] (76.center) to (66.center);
		\draw [style={alice_dashed}] (80.center) to (78.center);
		\draw [style={alice_dashed}] (79.center) to (77.center);
		\draw [style={alice_dashed}] (77.center) to (78.center);
		\draw [style={alice_dashed}] (80.center) to (79.center);
	\end{pgfonlayer}
\end{tikzpicture}

    \caption{Diagrammatic representation of $p_\gamma(\mathbf{t})$ from Eq.~\eqref{eq: post-processing behavior}. We start with a behavior $B$ that outputs joint outcomes $\mathbf{s}$. Using these outcomes, and some function $g$ that consistently maps $\mathbf{s} \in \mathcal{O}^\gamma$ to maximal contexts in $ \mathcal{C}_{POS}$ we input some choice $\delta=g(\mathbf{s})$ to perform measurements in $\mathcal{M}_{POS}$. Note that we average over all possible outcomes $\mathbf{s}$ to obtain $p_\gamma(\mathbf{t})$.}
    \label{fig: post-processing boxes}
\end{figure}
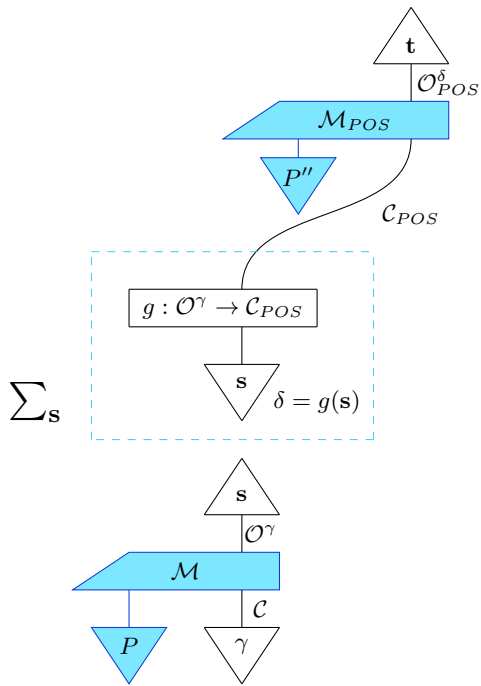

\subsubsection{Noncontextual wirings}

We can now define the free operations we will consider in this work: the  KS-\ti{noncontextual wirings} \cite{amaral2019resource}. We compose an arbitrary behavior $B$ with noncontextual pre-processing behaviors $B_{PRE} = \{\{p_\beta(\mathbf{r})\}_{\mathbf{r} \in \mathcal{O}^\beta}\}_{\beta \in \mathcal{C}} \in \Upsilon_{PRE}$ and a \emph{set} of adaptively chosen post-processing behaviors $\{\{B_{POS}^{(\mathbf{r},\beta)}\}_{\mathbf{r} \in \mathcal{O}^\beta}\}_{\beta \in \mathcal{C}}$. Importantly, we assume that the behaviors $B_{PRE}$ and $B_{POS}^{(\mathbf{r},\beta)}$ are noncontextual, for all $\mathbf{r}$ and $\beta$, otherwise, one could trivially create contextuality, and the operations would not be free. We will also restrict the structure of the adaptively chosen behaviors $B_{POS}^{(\mathbf{r},\beta)}$; they shall have a memory of the classical variables $\mathbf{r}$ or $\beta$ but remain factorizable. Moreover, we will also not allow the classical source of randomness $p(\lambda)$ (a. k. a. probability distribution sampling from some hidden-variable space, see Eq.~\eqref{eq: noncontextual factorizable}) from the KS-noncontextual behavior $B_{PRE}$ during the pre-processing to influence the choices of contexts during the post-processing. This is a choice made in Refs.~\cite{amaral2018noncontextual,amaral2019resource} that we will revisit later in our Results section in light of a discussion present in Ref.~\cite[App. A]{wolfe2020quantifying}.   

We show a diagrammatic representation of such a free operation in Fig.~\ref{fig: noncontextual wirings}. From left to right we increase the complexity of the wirings so that the rightmost operation is the most general noncontextual wiring operation. The colored boxes (in blue, color online) represent general behaviors (KS-noncontextual, quantum, and post-quantum). The boxes that are not colored represent only KS-noncontextual boxes, or classical inputs and outputs. Generic noncontextual wirings are divided into pre-processing and post-processing stages.

Let us describe the adaptivity, sometimes called a classical feed-forward operation,  allowed by the classical side channels. Each classical input $\beta$ and classical output $\mathbf{r}$ describing the behavior $B_{PRE}$ can be copied (white circles from Fig.~\ref{fig: noncontextual wirings}). Once copied, a side channel can transfer this information toward the post-processing stage, and depending on each such value, decide which post-processing operation to perform. This corresponds to a certain choice of behavior $B_{POS}^{(\mathbf{r},\beta)}$ -- conditioned on the classical feedforward of the values $\mathbf{r}$ and $\beta$ -- within a scenario $$\Upsilon_{POS}^{(\mathbf{r},\beta)} = (\mathcal{M}_{POS}^{(\mathbf{r},\beta)},\mathcal{C}_{POS}^{(\mathbf{r},\beta)},\mathcal{O}_{POS}^{(\mathbf{r},\beta)}).$$ To ease notation and description of the free operations (see  Fig.~\ref{fig: noncontextual wirings}) we can construct a large scenario $\Upsilon_{POS}$ by joining all possible $\Upsilon_{POS}^{(\mathbf{r},\beta)}$, 

\begin{equation*}
    \Upsilon_{POS} := \bigsqcup_{\beta \in \mathcal{C}_{PRE}} \bigsqcup_{\mathbf{r} \in \mathcal{O}^\beta} \Upsilon^{(\mathbf{r},\beta)}_{POS}.
\end{equation*}
Each element of $\Upsilon_{POS} = (\mathcal{M}_{POS},\mathcal{C}_{POS},\mathcal{O}_{POS})$ is the disjoint union of those pertaining to $\Upsilon^{(\mathbf{r},\beta)}_{POS}=(\mathcal{M}_{POS}^{(\mathbf{r},\beta)},\mathcal{C}_{POS}^{(\mathbf{r},\beta)},\mathcal{O}_{POS}^{(\mathbf{r},\beta)})$. This operation implies that each behavior $B$ from the scenario $\Upsilon_{POS}$ is the gluing of all behaviors in $\Upsilon^{(\mathbf{r},\beta)}$ but in a way that, for each choice $(\mathbf{r},\beta)$ only one behavior is considered, which has been termed the \emph{controlled choice}~\cite{abramsky2019comonadic} composition between scenarios (and behaviors). In terms of the associated convex polytope formed by $\mathrm{ND}(\Upsilon_{POS})$, the behaviors become $B_{POS} = ((B^{(\mathbf{r},\beta)})_{\mathbf{r} \in \mathcal{O}^\beta})_{\beta \in \mathcal{C}_{PRE}}$, that can be viewed as composing the convex polytopes $\mathrm{ND}(\Upsilon^{(\mathbf{r},\beta)}_{POS})$~\cite{wagner2023using}. 

As for when we have introduced post-processing operations, the wiring will need to adequate the outcomes  $\mathbf{s} \in \mathcal{O}^\gamma$ (related to $B$) to some maximal context $\mathcal{C}_{POS}^{(\mathbf{r},\beta)}$ within the scenario $\Upsilon_{POS}$ (see Fig.~\ref{fig: noncontextual wirings}). To that end, we now assume that there exists a family of functions $g(\cdot |\mathbf{r},\beta):\mathcal{O}^\gamma \to \mathcal{C}_{POS}^{(\mathbf{r},\beta)}$, and denote $\delta = g(\mathbf{s}|\mathbf{r},\beta)$. In words, the functions $g(\cdot|\mathbf{r},\beta)$ guarantee that the post-processing of $\mathbf{s}$ selects only compatible measurements to be jointly performed in the scenario $\Upsilon_{POS}$, depending on the pair $(\mathbf{r},\beta)$. See Fig.~\ref{fig: noncontextual wirings} for a description of how this can be operationally performed. Later we will discuss how general the function $g$ can be without being capable of creating contextuality out of noncontextual boxes.

Once the wirings happen, the outcomes of the final procedure will be some joint outcome $\mathbf{t} \in \mathcal{O}_{POS}^{\delta}$. Each output $\mathbf{t}$ of the post-processing box $B_{POS}$ can therefore be causally influenced by the inputs $\beta$ of $B_{PRE}$, and the outputs $\mathbf{r}$ of $B_{PRE}$~\cite{amaral2019resource}. However, this influence is not fully general, as commented in Ref.~\cite{amaral2018noncontextual}, we must demand that 

\begin{figure*}
    \centering
    \begin{tikzpicture}
	\begin{pgfonlayer}{nodelayer}
		\node [style=none] (0) at (2.25, 0) {};
		\node [style=none] (1) at (3.5, 0) {};
		\node [style=none] (2) at (3.5, -0.5) {};
		\node [style=none] (3) at (1.5, -0.5) {};
		\node [style=none] (9) at (2, -0.5) {};
		\node [style=none] (16) at (3, 0) {};
		\node [style=none] (17) at (3.25, -0.5) {};
		\node [style=none] (18) at (3, 0.5) {};
		\node [style=none] (22) at (2.5, 0.5) {};
		\node [style=none] (23) at (3.5, 0.5) {};
		\node [style=none] (24) at (3, 1.25) {};
		\node [style=none] (25) at (3, 0.75) {$\mathbf{s}$};
		\node [style=none] (27) at (2.75, -0.25) {$\mathcal{M}$};
		\node [style=none] (33) at (3.25, 0.25) {$\mathcal{O}^\gamma$};
		\node [style=none] (34) at (3.5, -0.75) {$\mathcal{C}$};
		\node [style=none] (45) at (4, 3.5) {};
		\node [style=none] (46) at (5.25, 3) {};
		\node [style=none] (47) at (4, 4) {};
		\node [style=none] (49) at (3.5, 4) {};
		\node [style=none] (50) at (4.5, 4) {};
		\node [style=none] (51) at (4, 4.75) {};
		\node [style=none] (52) at (4, 4.25) {$\mathbf{t}$};
		\node [style=none] (54) at (4.25, 3.25) {$\mathcal{M}_{POS}$};
		\node [style=none] (59) at (4.5, 3.75) {$\mathcal{O}^\delta_{POS}$};
		\node [style=none] (60) at (4.75, 2.75) {$\mathcal{C}_{POS}$};
		\node [style=none] (67) at (3.25, 1.25) {};
		\node [style=none] (71) at (3.5, 2) {};
		\node [style=none] (73) at (3.5, 1.5) {};
		\node [style=none] (112) at (3, -4) {};
		\node [style=none] (113) at (3.75, -3.5) {};
		\node [style=none] (114) at (4.25, -4) {};
		\node [style=none] (115) at (3.75, -2.75) {};
		\node [style=none] (116) at (4.25, -4.5) {};
		\node [style=none] (117) at (3.25, -2.75) {};
		\node [style=none] (118) at (4.25, -2.75) {};
		\node [style=none] (119) at (3.75, -2) {};
		\node [style=none] (120) at (3.75, -2.5) {$\mathbf{r}$};
		\node [style=none] (122) at (3.75, -3.75) {$\mathcal{M}_{PRE}$};
		\node [style=none] (123) at (3.75, -4.5) {};
		\node [style=none] (124) at (4.75, -4.5) {};
		\node [style=none] (125) at (4.25, -5.25) {};
		\node [style=none] (126) at (4.25, -4.75) {$\beta$};
		\node [style=none] (127) at (3.25, -3.25) {$\mathcal{O}^\beta_{PRE}$};
		\node [style=none] (128) at (4.75, -4.25) {$\mathcal{C}_{PRE}$};
		\node [style=none] (129) at (3.25, -1) {};
		\node [style=none] (130) at (3.75, -1) {};
		\node [style=none] (131) at (2.75, -1) {};
		\node [style=none] (132) at (3.25, -1.75) {};
		\node [style=none] (133) at (3.25, -1.25) {$\gamma$};
		\node [style=none] (149) at (2.75, 3) {};
		\node [style=none] (150) at (3.25, 3.5) {};
		\node [style=none] (151) at (5.75, 3.5) {};
		\node [style=none] (152) at (5.75, 3) {};
		\node [style=none] (153) at (3.25, 3) {};
		\node [style=none] (157) at (2.25, -4) {};
		\node [style=none] (158) at (2.75, -3.5) {};
		\node [style=none] (159) at (5, -3.5) {};
		\node [style=none] (160) at (5, -4) {};
		\node [style=none] (167) at (3.25, 2.75) {};
		\node [style=none] (168) at (2.75, 2.75) {};
		\node [style=none] (169) at (3.75, 2.75) {};
		\node [style=none] (170) at (3.25, 2) {};
		\node [style=none] (171) at (3.25, 2.5) {$\phi$};
		\node [style=none] (172) at (3, -4.5) {};
		\node [style=none] (173) at (2.5, -4.5) {};
		\node [style=none] (174) at (3.5, -4.5) {};
		\node [style=none] (175) at (3, -5.25) {};
		\node [style=none] (176) at (3, -4.75) {$\lambda$};
		\node [style=none] (177) at (1.75, -4.75) {$\sum_\lambda p(\lambda)$};
		\node [style=none] (178) at (2, 2.5) {$\sum_\phi p(\phi)$};
		\node [style=none] (179) at (1.5, -1) {};
		\node [style=none] (180) at (2.5, -1) {};
		\node [style=none] (181) at (2, -1.75) {};
		\node [style=none] (182) at (2, -1.25) {$P$};
		\node [style=none] (183) at (2, -1) {};
		\node [style=none] (184) at (7.75, 0) {};
		\node [style=none] (185) at (9, 0) {};
		\node [style=none] (186) at (9, -0.5) {};
		\node [style=none] (187) at (7, -0.5) {};
		\node [style=none] (188) at (7.5, -0.5) {};
		\node [style=none] (189) at (8.5, 0) {};
		\node [style=none] (190) at (8.5, -0.5) {};
		\node [style=none] (191) at (8.5, 0.5) {};
		\node [style=none] (192) at (8, 0.5) {};
		\node [style=none] (193) at (9, 0.5) {};
		\node [style=none] (194) at (8.5, 1.25) {};
		\node [style=none] (195) at (8.5, 0.75) {$\mathbf{s}$};
		\node [style=none] (196) at (8.25, -0.25) {$\mathcal{M}$};
		\node [style=none] (197) at (8.75, 0.25) {$\mathcal{O}^\gamma$};
		\node [style=none] (198) at (8.75, -0.75) {$\mathcal{C}$};
		\node [style=none] (199) at (9.5, 3.5) {};
		\node [style=none] (200) at (10.75, 3) {};
		\node [style=none] (201) at (9.5, 4) {};
		\node [style=none] (203) at (9, 4) {};
		\node [style=none] (204) at (10, 4) {};
		\node [style=none] (205) at (9.5, 4.75) {};
		\node [style=none] (206) at (9.5, 4.25) {$\mathbf{t}$};
		\node [style=none] (207) at (9.75, 3.25) {$\mathcal{M}_{POS}$};
		\node [style=none] (208) at (10, 3.75) {$\mathcal{O}^\delta_{POS}$};
		\node [style=none] (209) at (10.25, 2.75) {$\mathcal{C}_{POS}$};
		\node [style=none] (211) at (8.75, 1.25) {};
		\node [style=none] (213) at (9.25, 0.5) {};
		\node [style=none] (221) at (8.5, -4) {};
		\node [style=none] (222) at (9.25, -3.5) {};
		\node [style=none] (223) at (9.75, -4) {};
		\node [style=none] (224) at (9.25, -2.75) {};
		\node [style=none] (225) at (9.75, -4.5) {};
		\node [style=none] (226) at (8.75, -2.75) {};
		\node [style=none] (227) at (9.75, -2.75) {};
		\node [style=none] (228) at (9.25, -2) {};
		\node [style=none] (229) at (9.25, -2.5) {$\mathbf{r}$};
		\node [style=none] (230) at (9.25, -3.75) {$\mathcal{M}_{PRE}$};
		\node [style=none] (231) at (9.25, -4.5) {};
		\node [style=none] (232) at (10.25, -4.5) {};
		\node [style=none] (233) at (9.75, -5.25) {};
		\node [style=none] (234) at (9.75, -4.75) {$\beta$};
		\node [style=none] (235) at (8.75, -3.25) {$\mathcal{O}^\beta_{PRE}$};
		\node [style=none] (236) at (10.25, -4.25) {$\mathcal{C}_{PRE}$};
		\node [style=none] (237) at (8.5, -2) {};
		\node [style=none] (238) at (9, -2) {};
		\node [style=none] (239) at (8, -2) {};
		\node [style=none] (240) at (8.5, -2.75) {};
		\node [style=none] (241) at (8.5, -2.25) {$\mathbf{r}$};
		\node [style=none] (242) at (8.25, -1) {};
		\node [style=none] (243) at (8.75, -1) {};
		\node [style=none] (244) at (8.25, -1.5) {};
		\node [style=none] (245) at (8.75, -1.5) {};
		\node [style=none] (246) at (8.5, -1.5) {};
		\node [style=none] (247) at (8.5, -1) {};
		\node [style=none] (248) at (8.5, -1.25) {$f$};
		\node [style=none] (249) at (8.25, 3) {};
		\node [style=none] (250) at (8.75, 3.5) {};
		\node [style=none] (251) at (11.25, 3.5) {};
		\node [style=none] (252) at (11.25, 3) {};
		\node [style=none] (253) at (8.75, 3) {};
		\node [style=none] (254) at (7.75, -4) {};
		\node [style=none] (255) at (8.25, -3.5) {};
		\node [style=none] (256) at (10.5, -3.5) {};
		\node [style=none] (257) at (10.5, -4) {};
		\node [style=none] (261) at (11.25, -4) {};
		\node [style=none] (263) at (8.75, 2.75) {};
		\node [style=none] (264) at (8.25, 2.75) {};
		\node [style=none] (265) at (9.25, 2.75) {};
		\node [style=none] (266) at (8.75, 2) {};
		\node [style=none] (267) at (8.75, 2.5) {$\phi$};
		\node [style=none] (268) at (8.5, -4.5) {};
		\node [style=none] (269) at (8, -4.5) {};
		\node [style=none] (270) at (9, -4.5) {};
		\node [style=none] (271) at (8.5, -5.25) {};
		\node [style=none] (272) at (8.5, -4.75) {$\lambda$};
		\node [style=none] (273) at (7.25, -4.75) {$\sum_\lambda p(\lambda)$};
		\node [style=none] (274) at (7.5, 2.5) {$\sum_\phi p(\phi)$};
		\node [style=none] (275) at (7, -1) {};
		\node [style=none] (276) at (8, -1) {};
		\node [style=none] (277) at (7.5, -1.75) {};
		\node [style=none] (278) at (7.5, -1.25) {$P$};
		\node [style=none] (279) at (7.5, -1) {};
		\node [style=none] (280) at (13.25, 0) {};
		\node [style=none] (281) at (14.5, 0) {};
		\node [style=none] (282) at (14.5, -0.5) {};
		\node [style=none] (283) at (12.5, -0.5) {};
		\node [style=none] (284) at (13, -0.5) {};
		\node [style=none] (285) at (14, 0) {};
		\node [style=none] (286) at (14, -0.5) {};
		\node [style=none] (287) at (14, 0.5) {};
		\node [style=none] (288) at (13.5, 0.5) {};
		\node [style=none] (289) at (14.5, 0.5) {};
		\node [style=none] (290) at (14, 1.25) {};
		\node [style=none] (291) at (14, 0.75) {$\mathbf{s}$};
		\node [style=none] (292) at (13.75, -0.25) {$\mathcal{M}$};
		\node [style=none] (293) at (14.25, 0.25) {$\mathcal{O}^\gamma$};
		\node [style=none] (294) at (14.25, -0.75) {$\mathcal{C}$};
		\node [style=none] (295) at (15, 3.5) {};
		\node [style=none] (296) at (15.75, 3) {};
		\node [style=none] (297) at (15, 4) {};
		\node [style=none] (298) at (15.75, 2) {};
		\node [style=none] (299) at (14.5, 4) {};
		\node [style=none] (300) at (15.5, 4) {};
		\node [style=none] (301) at (15, 4.75) {};
		\node [style=none] (302) at (15, 4.25) {$\mathbf{t}$};
		\node [style=none] (303) at (15.25, 3.25) {$\mathcal{M}_{POS}$};
		\node [style=none] (304) at (15.5, 3.75) {$\mathcal{O}^\delta_{POS}$};
		\node [style=none] (305) at (15.25, 2.5) {$\mathcal{C}_{POS}$};
		\node [style=none] (306) at (14.75, 1.25) {};
		\node [style=none] (307) at (14.25, 1.25) {};
		\node [style=none] (308) at (15.25, 1.25) {};
		\node [style=none] (309) at (14.75, 0.5) {};
		\node [style=none] (310) at (14.75, 1) {$\mathbf{s}$};
		\node [style=none] (311) at (14.5, 2) {};
		\node [style=none] (312) at (17, 2) {};
		\node [style=none] (313) at (14.5, 1.5) {};
		\node [style=none] (314) at (17, 1.5) {};
		\node [style=none] (315) at (15.75, 1.75) {$g$};
		\node [style=none] (316) at (14.75, 1.5) {};
		\node [style=none] (317) at (14, -4) {};
		\node [style=none] (318) at (14.75, -3.5) {};
		\node [style=none] (319) at (15.25, -4) {};
		\node [style=none] (320) at (14.75, -2.75) {};
		\node [style=none] (321) at (15.25, -4.75) {};
		\node [style=none] (322) at (14.25, -2.75) {};
		\node [style=none] (323) at (15.25, -2.75) {};
		\node [style=none] (324) at (14.75, -2) {};
		\node [style=none] (325) at (14.75, -2.5) {$\mathbf{r}$};
		\node [style=none] (326) at (14.75, -3.75) {$\mathcal{M}_{PRE}$};
		\node [style=none] (327) at (14.75, -4.75) {};
		\node [style=none] (328) at (15.75, -4.75) {};
		\node [style=none] (329) at (15.25, -5.5) {};
		\node [style=none] (330) at (15.25, -5) {$\beta$};
		\node [style=none] (331) at (14.25, -3.25) {$\mathcal{O}^\beta_{PRE}$};
		\node [style=none] (332) at (15.75, -4.25) {$\mathcal{C}_{PRE}$};
		\node [style=none] (333) at (14, -2) {};
		\node [style=none] (334) at (14.5, -2) {};
		\node [style=none] (335) at (13.5, -2) {};
		\node [style=none] (336) at (14, -2.75) {};
		\node [style=none] (337) at (14, -2.25) {$\mathbf{r}$};
		\node [style=none] (338) at (13.75, -1) {};
		\node [style=none] (339) at (14.25, -1) {};
		\node [style=none] (340) at (13.75, -1.5) {};
		\node [style=none] (341) at (14.25, -1.5) {};
		\node [style=none] (342) at (14, -1.5) {};
		\node [style=none] (343) at (14, -1) {};
		\node [style=none] (344) at (14, -1.25) {$f$};
		\node [style=none] (345) at (13.75, 3) {};
		\node [style=none] (346) at (14.25, 3.5) {};
		\node [style=none] (347) at (16.75, 3.5) {};
		\node [style=none] (348) at (16.75, 3) {};
		\node [style=none] (349) at (14.25, 3) {};
		\node [style=none] (350) at (13.25, -4) {};
		\node [style=none] (351) at (13.75, -3.5) {};
		\node [style=none] (352) at (16, -3.5) {};
		\node [style=none] (353) at (16, -4) {};
		\node [style=none] (355) at (16, 1.5) {};
		\node [style=copy] (356) at (15.25, -4.5) {};
		\node [style=none] (357) at (16.75, -4) {};
		\node [style=none] (358) at (16.75, 1.5) {};
		\node [style=none] (359) at (14.25, 2.75) {};
		\node [style=none] (360) at (13.75, 2.75) {};
		\node [style=none] (361) at (14.75, 2.75) {};
		\node [style=none] (362) at (14.25, 2) {};
		\node [style=none] (363) at (14.25, 2.5) {$\phi$};
		\node [style=none] (364) at (14, -4.5) {};
		\node [style=none] (365) at (13.5, -4.5) {};
		\node [style=none] (366) at (14.5, -4.5) {};
		\node [style=none] (367) at (14, -5.25) {};
		\node [style=none] (368) at (14, -4.75) {$\lambda$};
		\node [style=none] (369) at (12.75, -4.75) {$\sum_\lambda p(\lambda)$};
		\node [style=none] (370) at (13, 2.5) {$\sum_\phi p(\phi)$};
		\node [style=none] (371) at (12.5, -1) {};
		\node [style=none] (372) at (13.5, -1) {};
		\node [style=none] (373) at (13, -1.75) {};
		\node [style=none] (374) at (13, -1.25) {$P$};
		\node [style=none] (375) at (13, -1) {};
		\node [style=none] (377) at (5.75, 2.25) {};
		\node [style=none] (378) at (4.75, 2.25) {};
		\node [style=none] (379) at (5.25, 1.5) {};
		\node [style=none] (380) at (5.25, 2.25) {};
		\node [style=none] (381) at (10.75, 2.25) {};
		\node [style=none] (382) at (11.25, 2.25) {};
		\node [style=none] (383) at (10.25, 2.25) {};
		\node [style=none] (384) at (10.75, 1.5) {};
		\node [style=none] (385) at (10.75, 2) {$\delta$};
		\node [style=none] (386) at (5.25, 2) {$\delta$};
		\node [style=copy] (387) at (14.75, -3) {};
		\node [style=none] (388) at (16, -2.5) {};
	\end{pgfonlayer}
	\begin{pgfonlayer}{edgelayer}
		\draw [style=alice] (1.center)
			 to (0.center)
			 to (3.center)
			 to (2.center)
			 to cycle;
		\draw (16.center) to (18.center);
		\draw (23.center) to (22.center);
		\draw (22.center) to (24.center);
		\draw (24.center) to (23.center);
		\draw (45.center) to (47.center);
		\draw (50.center) to (49.center);
		\draw (49.center) to (51.center);
		\draw (51.center) to (50.center);
		\draw (113.center) to (115.center);
		\draw (114.center) to (116.center);
		\draw (118.center) to (117.center);
		\draw (117.center) to (119.center);
		\draw (119.center) to (118.center);
		\draw (123.center) to (124.center);
		\draw (123.center) to (125.center);
		\draw (125.center) to (124.center);
		\draw (131.center) to (130.center);
		\draw (130.center) to (132.center);
		\draw (132.center) to (131.center);
		\draw (152.center) to (149.center);
		\draw (149.center) to (150.center);
		\draw (150.center) to (151.center);
		\draw (151.center) to (152.center);
		\draw (157.center) to (160.center);
		\draw (160.center) to (159.center);
		\draw (159.center) to (158.center);
		\draw (158.center) to (157.center);
		\draw (168.center) to (169.center);
		\draw (168.center) to (170.center);
		\draw (170.center) to (169.center);
		\draw (167.center) to (153.center);
		\draw (173.center) to (174.center);
		\draw (173.center) to (175.center);
		\draw (175.center) to (174.center);
		\draw (172.center) to (112.center);
		\draw [style=alice] (179.center)
			 to (181.center)
			 to (180.center)
			 to cycle;
		\draw [style=alice] (183.center) to (9.center);
		\draw [style=alice] (185.center)
			 to (184.center)
			 to (187.center)
			 to (186.center)
			 to cycle;
		\draw (189.center) to (191.center);
		\draw (193.center) to (192.center);
		\draw (192.center) to (194.center);
		\draw (194.center) to (193.center);
		\draw (199.center) to (201.center);
		\draw (204.center) to (203.center);
		\draw (203.center) to (205.center);
		\draw (205.center) to (204.center);
		\draw (222.center) to (224.center);
		\draw (223.center) to (225.center);
		\draw (227.center) to (226.center);
		\draw (226.center) to (228.center);
		\draw (228.center) to (227.center);
		\draw (231.center) to (232.center);
		\draw (231.center) to (233.center);
		\draw (233.center) to (232.center);
		\draw (239.center) to (238.center);
		\draw (238.center) to (240.center);
		\draw (240.center) to (239.center);
		\draw (244.center) to (245.center);
		\draw (245.center) to (243.center);
		\draw (243.center) to (242.center);
		\draw (242.center) to (244.center);
		\draw [in=-90, out=90, looseness=0.75] (237.center) to (246.center);
		\draw [in=-90, out=90, looseness=0.75] (247.center) to (190.center);
		\draw (252.center) to (249.center);
		\draw (249.center) to (250.center);
		\draw (250.center) to (251.center);
		\draw (251.center) to (252.center);
		\draw (254.center) to (257.center);
		\draw (257.center) to (256.center);
		\draw (256.center) to (255.center);
		\draw (255.center) to (254.center);
		\draw (264.center) to (265.center);
		\draw (264.center) to (266.center);
		\draw (266.center) to (265.center);
		\draw (263.center) to (253.center);
		\draw (269.center) to (270.center);
		\draw (269.center) to (271.center);
		\draw (271.center) to (270.center);
		\draw (268.center) to (221.center);
		\draw [style=alice] (275.center)
			 to (277.center)
			 to (276.center)
			 to cycle;
		\draw [style=alice] (279.center) to (188.center);
		\draw [style=alice] (281.center)
			 to (280.center)
			 to (283.center)
			 to (282.center)
			 to cycle;
		\draw (285.center) to (287.center);
		\draw (289.center) to (288.center);
		\draw (288.center) to (290.center);
		\draw (290.center) to (289.center);
		\draw (295.center) to (297.center);
		\draw [in=90, out=-90] (296.center) to (298.center);
		\draw (300.center) to (299.center);
		\draw (299.center) to (301.center);
		\draw (301.center) to (300.center);
		\draw (307.center) to (308.center);
		\draw (307.center) to (309.center);
		\draw (309.center) to (308.center);
		\draw (313.center) to (311.center);
		\draw (311.center) to (312.center);
		\draw (312.center) to (314.center);
		\draw (314.center) to (313.center);
		\draw [in=90, out=-90, looseness=0.75] (316.center) to (306.center);
		\draw (318.center) to (320.center);
		\draw (319.center) to (321.center);
		\draw (323.center) to (322.center);
		\draw (322.center) to (324.center);
		\draw (324.center) to (323.center);
		\draw (327.center) to (328.center);
		\draw (327.center) to (329.center);
		\draw (329.center) to (328.center);
		\draw (335.center) to (334.center);
		\draw (334.center) to (336.center);
		\draw (336.center) to (335.center);
		\draw (340.center) to (341.center);
		\draw (341.center) to (339.center);
		\draw (339.center) to (338.center);
		\draw (338.center) to (340.center);
		\draw [in=-90, out=90, looseness=0.75] (333.center) to (342.center);
		\draw [in=-90, out=90, looseness=0.75] (343.center) to (286.center);
		\draw (348.center) to (345.center);
		\draw (345.center) to (346.center);
		\draw (346.center) to (347.center);
		\draw (347.center) to (348.center);
		\draw (350.center) to (353.center);
		\draw (353.center) to (352.center);
		\draw (352.center) to (351.center);
		\draw (351.center) to (350.center);
		\draw [in=-90, out=0] (356) to (357.center);
		\draw (357.center) to (358.center);
		\draw (360.center) to (361.center);
		\draw (360.center) to (362.center);
		\draw (362.center) to (361.center);
		\draw (359.center) to (349.center);
		\draw (365.center) to (366.center);
		\draw (365.center) to (367.center);
		\draw (367.center) to (366.center);
		\draw (364.center) to (317.center);
		\draw [style=alice] (371.center)
			 to (373.center)
			 to (372.center)
			 to cycle;
		\draw [style=alice] (375.center) to (284.center);
		\draw (129.center) to (17.center);
		\draw (378.center) to (377.center);
		\draw (377.center) to (379.center);
		\draw (379.center) to (378.center);
		\draw (380.center) to (46.center);
		\draw (383.center) to (382.center);
		\draw (382.center) to (384.center);
		\draw (384.center) to (383.center);
		\draw (381.center) to (200.center);
		\draw [in=-90, out=0, looseness=1.25] (387) to (388.center);
		\draw (388.center) to (355.center);
	\end{pgfonlayer}
\end{tikzpicture}

    \caption{Forming a noncontextual wiring. From left to right we increase wiring `complexity', towards a completely general noncontextual wiring. (Left) We start with three independent statistics, in three independent scenarios. Two are noncontextual behaviors and one is a general behavior. (Middle) For each result $\mathbf{r}$, occurring with probability $p_\beta(\mathbf{r}) \in B_{PRE}$ a noncontextual behavior, we use a function $f$ to output a valid context $\gamma = f(\mathbf{r})$ in $\mathcal{C}$. (Right) Similarly to the pre-processing, for each outcome $\mathbf{s}$ a function $g$ outputs a valid context $\delta = g(\mathbf{s})$ in $\mathcal{C}_{POS} \equiv \mathcal{C}_{POS}^{(\mathbf{r},\beta)}$ that can now be adaptively chosen. Depending on the joint outcomes $\mathbf{r}$ of the pre-selected behavior and maximal contexts $\beta$ of the pre-selected scenario we may choose any noncontextual behavior $B_{POS}^{(\mathbf{r},\beta)}$ within a scenario $\Upsilon^{(\mathbf{r},\beta)}_{POS}$ with the condition that knowledge of $\mathbf{r}$ and $\beta$ preserves the factorizability of the noncontextual behavior, as described by Eq.~\eqref{eq: adaptively respecting factorizability}. For clarity of the presentation, we have omitted summations  $\sum_{\mathbf{r},\mathbf{s}}$. As before, the causal order of the operations goes from bottom (input context $\beta$) to top (joint output $\mathbf{t}$). }
    \label{fig: noncontextual wirings}
\end{figure*}
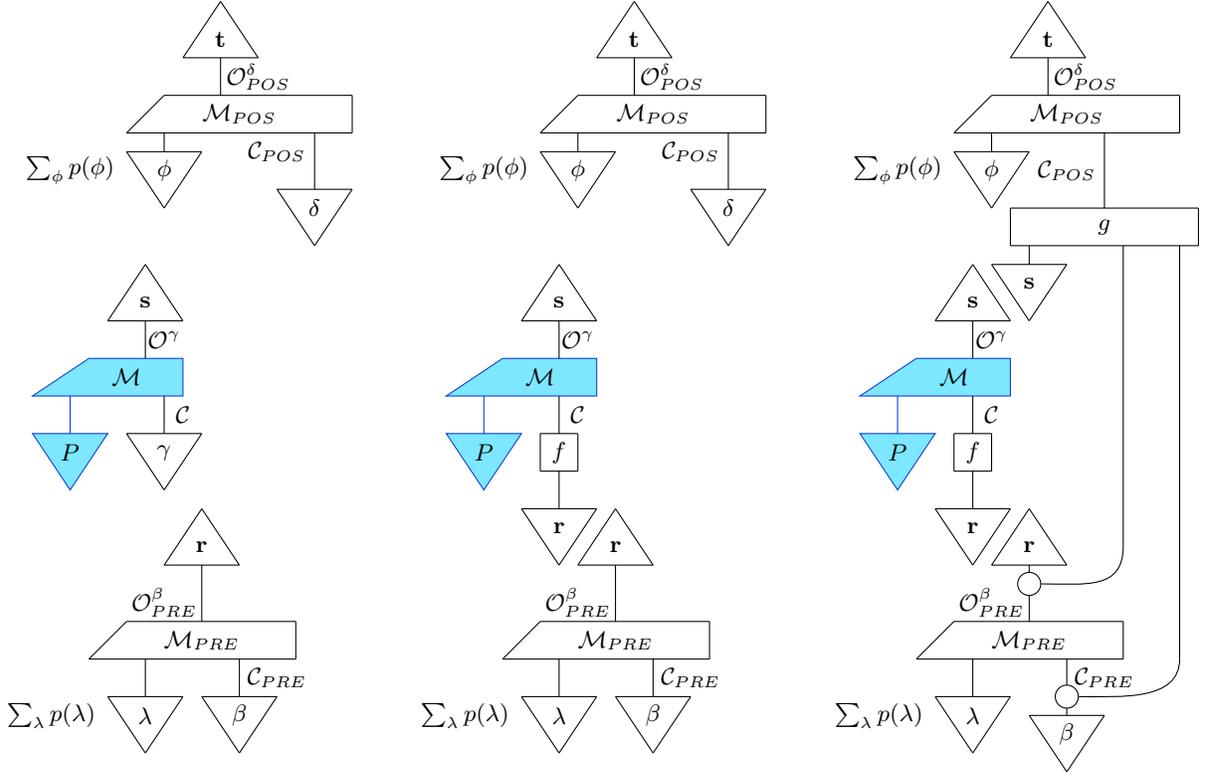

\begin{equation}\label{eq: adaptively respecting factorizability}
    p_{\delta} (\tb{t} |\tb{s}, \tb{r},\beta) =  \sum_\phi p(\phi) \prod_{\delta_i \in \delta = g(\mathbf{s}|\mathbf{r},\beta)} p_{\delta_i} (t_i | \phi).
\end{equation}
The above restriction translates, in words, to the fact that $B_{POS}$ remains a KS-noncontextual behavior, even with the adaptivity from the classical inputs $\beta$ and $\mathbf{r}$. The behavior $B_{POS}$ must be described by a factorizable hidden-variable model, and the function $g$ cannot use the information of $\mathbf{r}$ and $\beta$ to change that. Recall that a noncontextual factorizable model is of the form given by Eq.~\eqref{eq: noncontextual factorizable}, where the factorizability corresponds to the fact that the total probability distribution of a context $\gamma$ is described by a source of classical randomness $p(\lambda)$ and the product of independent distributions $p_{\gamma_i}(s_i|\lambda)$ for each element $\gamma_i$ of the context $\gamma$, and $s_i$ the associated outcome of measurement $\gamma_i$. Eq.~\eqref{eq: adaptively respecting factorizability} requires that the same structure remains valid, for every $\delta_i \in \delta = g(\mathbf{s}|\mathbf{r},\beta)$ and outcomes $t_i$ of this measurement $\delta_i$. Even though we can now use the feedforward classical information to choose different scenarios, and consequently behaviors therein, this structure remains valid.

With this construction, we get the final scenario once  wirings take place $$\Upsilon_{W} := (\C{M}_{PRE}, \C{C}_{PRE}, \C{O}_{POS}),$$ where behaviors $\C{W}(B) \in \Upsilon_{W}$, are the set of all probabilities of the form

\begin{equation}\label{eq: Noncontextual wiring behavior}
    p_\beta (\tb{t}) = \sum_{\tb{r} \in \mathcal{O}^\beta,\tb{s}\in \mathcal{O}^\gamma} p_\beta (\tb{r}) p_{\gamma=f(\tb{r})} \left(\tb{s}\right) p_{\delta=g(\mathbf{s}|\mathbf{r},\beta)}\left(\tb{t}\right),
\end{equation}
where $\beta \in \mathcal{C}_{PRE}$ and $\mathbf{t} \in \mathcal{O}_{POS}^{(\mathbf{r},\beta)}$ for some choice of pre-processing input $\beta$ and the successful pre-processing output $\mathbf{r}$. Importantly, $p_\beta(\mathbf{r})$ has the form given by Eq.~\eqref{eq: noncontextual factorizable} and $p_{\delta=g(\mathbf{s}|\mathbf{r},\beta)}(\mathbf{t})$ has the form given by Eq.~\eqref{eq: adaptively respecting factorizability}. The behavior $B=\{p_\gamma(\mathbf{s})\}$ can be any non-disturbing behavior. For future reference, we collect what we have discussed so far in the following definition:

\begin{definition}[Noncontextual wirings]\label{def: noncontextual wirings}
    Let $\Upsilon = (\mathcal{M},\mathcal{C},\mathcal{O}^{\mathcal{M}})$ be a compatibility scenario. Let also $\Upsilon_{PRE} = (\mathcal{M}_{PRE},\mathcal{C}_{PRE},\mathcal{O}_{PRE})$ be any (pre-processing) compatibility scenario, and for all $(\mathbf{r},\beta) \in \mathcal{O}_{PRE}^\beta \times \mathcal{C}_{PRE}$, let  $\Upsilon_{POS}^{(\mathbf{r},\beta)} = (\mathcal{M}_{POS}^{(\mathbf{r},\beta)},\mathcal{C}_{POS}^{(\mathbf{r},\beta)},\mathcal{O}_{POS}^{(\mathbf{r},\beta)})$ be also generic (collections of post-processing) compatibility scenarios. A noncontextual wiring is a mapping $\mathcal{W}: \Upsilon \to \Upsilon_W$ defined via Eq.~\eqref{eq: Noncontextual wiring behavior}, satisfying: 
    \begin{enumerate}
        \item $\{p_\beta(\mathbf{r})\}$ is a noncontextual behavior in $\Upsilon_{PRE}$.
        \item \{$p_{\delta = g(\mathbf{s}|\mathbf{r},\beta)}(\mathbf{t})\}$ is a noncontextual behavior in $\Upsilon_{POS}^{(\mathbf{r},\beta)}$, for all $\mathbf{r},\beta$.
        \item The function $g(\cdot|\mathbf{r},\beta): \mathcal{O}^\gamma \to \mathcal{C}_{POS}^{(\mathbf{r},\beta)}$, for each $\gamma \in \mathcal{C}$, must satisfy $g(\mathbf{s}|\mathbf{r},\beta) = \cup_i g_i(s_i|\mathbf{r},\beta)$, i.e., each outcome $s_i \in \mathbf{s}$ of a measurement $\gamma_i \in \gamma$ is associated independently to a given measurement in $\mathcal{M}_{POS}^{(\mathbf{r},\beta)}$ by some function $g_i(\cdot |\mathbf{r},\beta): \mathcal{O}^{\gamma_i} \to \mathcal{M}_{POS}^{(\mathbf{r},\beta))}$ such that $g(\mathbf{s}|\mathbf{r},\beta) \in \mathcal{C}_{POS}^{(\mathbf{r},\beta)}$.
    \end{enumerate}
    Where $\Gamma_W = (\mathcal{M}_{PRE},\mathcal{C}_{PRE},\mathcal{O}_{POS})$, with $\mathcal{O}_{POS}$ given by the disjoint union of all $\mathcal{O}_{POS}^{(\mathbf{r},\beta)}$.
\end{definition}

This particular class of operations will be denoted as $\mathrm{NCW}$. Definition~\ref{def: noncontextual wirings} is involved, yet we hope Fig.~\ref{fig: noncontextual wirings} can provide an intuitive diagrammatic representation. The restriction on the possible choices of $g$ is important so that choices of contexts during the post-processing stage cannot be determined by non-local influences (as will be discussed in Sec.~\ref{sec: NCW vs LOSR}, this implies that in a bipartite Bell scenario $g$ `breaks' into two boxes $g_{\mathrm A}$ and $g_{\mathrm B}$ for each party, see Figs.~\ref{fig: LOSR^circ} and~\ref{fig: LOSR^bullet}).   As was shown in Ref.~\cite{amaral2019resource}, if $B \in \mathrm{ND}(\Upsilon)$ then $\mathcal{W}(B) \in \mathrm{ND}(\Upsilon_W)$. More importantly, operations in $\mathrm{NCW}$ preserve KS-noncontextuality, i.e.,  if $B \in \mathrm{NC}(\Upsilon)$ then $\mathcal{W}(B) \in \mathrm{NC}(\Upsilon_W)$. The resource theory $\mathcal{R}_{\mathrm{NCW}} = (\mathcal{U}_{\mathrm{ND}},\mathcal{F}_{\mathrm{NC}},\mathrm{NCW})$ is the resource theory of KS-noncontextuality where free operations are taken to be noncontextual wirings.

We call the \emph{type} of a box $B$ to be the scenario $\Upsilon = (\C{M}, \C{C}, \C{O}^{\C{M}})$ with respect to which the box $B$ is defined. This terminology is introduced because the same scenario has many different behaviors (possible non-disturbing boxes) all having the same `type' (i.e., all having the same compatibility structure for maximal contexts, number of measurements, and number of joint outcomes per context). Now, as our NCW operations take boxes in $\Upsilon = (\C{M}, \C{C}, \C{O}^{\C{M}})$ to boxes in $\Upsilon_{W} = (\C{M}_{PRE}, \C{C}_{PRE}, \C{O}_{POS})$, we define the {\ti{type of an operation}} $\C{W}$ as the pair of input/output types of scenarios $\Upsilon \rightarrow \Upsilon_W$. The set of all noncontextual wiring operations of type $\Upsilon_{\mathrm{in}} \rightarrow \Upsilon_{\mathrm{out}}$ is denoted by $\mathrm{NCW}(\Upsilon_{\mathrm{in}} \rightarrow \Upsilon_{\mathrm{out}})$. As mentioned above, some of our results will be type-specific, meaning results concerning type-dependent operations~\footnote{It is useful, for clarity and consistency, to distinguish our usage of `type' from others encountered in the literature. For example, Refs.~\cite{schmid2021unscrambling} use the word `type' to denote the systems considered (classical, quantum or GPT) within a given process theory. Hence, a type labels a different system. Ref.~\cite{barbosa2023closing} considers a `type' to be the same as considered here, i.e., the compatibility scenario (that they term measurement scenario), and equivalently the type of a simulation to be the pair of scenarios in the domain and codomain of a free operation. Refs.~\cite{rosset2020type} considered `type' to denote input/output structure of choices of systems of a given process, e.g. a generic bipartite channel has the type $\mathsf{QQ} \to \mathsf{QQ}$ because it has inputs and outputs two quantum systems, while a bipartite shared randomness has the type $\mathsf{II} \to \mathsf{CC}$ that inputs the trivial system and outputs two classical systems. Finally, Ref.~\cite{wolfe2020quantifying} used `type' to denote the prescription of the number of inputs/outputs in a bipartite Bell scenario, for example, they write $\left(\begin{matrix}
    2 & 2 \\ 2 & 2
\end{matrix}\right)$ corresponding to two labels as inputs/outputs per `wing' in a Bell causal structure. The whole matrix is called the `type' of the specific common cause scenario.}.

One elementary local comparability result is the following one:
\begin{lemma}\label{Lemma: noncontextual wiring always returning the same box}
    For every $B \in \mathrm{ND}(\Upsilon)\setminus \mathrm{NC}(\Upsilon)$ and every $B_{NC} \in \mathrm{NC}(\Upsilon)$ there exists some noncontextual wiring $\mathcal{W}$ such that $B_{NC} = \mathcal{W}(B)$, i.e., $B \to B_{NC}$.
\end{lemma}
\begin{proof}
    We can choose a noncontextual wiring that is merely a post-processing where we discard the outcomes from $B$ and for which the post-processing behavior is given by $B_{NC}$.
\end{proof}

\subsubsection{Deterministic noncontextual wirings}\label{sec: noncontextual deterministic}

Now we look for a special class of operations, the set of \ti{deterministic noncontextual wirings}. Consider a pre-processing behavior

\be
    p_{\beta} (\tb{r}) =  \sum_\phi p(\phi) \prod_{i} p_{\beta_i} (r_i | \phi).
\ee
It is well known that, without loss of generality, we can assume $p_{\beta_i}$ to take values in $\{0,1\}$ for all $i$~\cite{barbosa2022continuous}. This is possible because we can allow any source of randomness to be due to the source $p(\phi)$. In case such values are deterministically assigned, no randomness $p(\phi)$ is considered, and the behavior is said to be deterministic and takes the form

\begin{equation}\label{eq: noncontextual deterministic}
    d_\beta (\tb{r}) = \prod_i d_{\beta_i} (r_i).
\end{equation}

Above, $d_{\beta_i}(r_i) \in \{0,1\}$, for all $i$. Hence, in a deterministic behavior, each context selects a unique output string $\tb{r}$.  This allows us to formulate a definition for a deterministic noncontextual wiring operation. 

\begin{definition}[Deterministic noncontextual wirings]
 We say that a noncontextual wiring operation is \emph{deterministic} when both pre- and post-processings are constructed using noncontextual deterministic behaviors. 
\end{definition}

\section{Results}

\subsection{Restriction over Bell scenarios}\label{sec: NCW vs LOSR}

We start our results section by analyzing which set of free transformations noncontextual wirings become when we restrict noncontextual scenarios to be those where the compatibility structure is isomorphic to some Bell scenario~\cite{brunner2014bell}. Recall that, albeit different (both conceptually and operationally) notions of nonclassicality,  every Bell scenario can be viewed as a compatibility scenario and every behavior in a Bell scenario can be viewed as a behavior in that compatibility scenario~\cite{abramsky2011sheaf,abramsky2016possibilities,abramsky2012cohomology}. Formally, the no-disturbance condition is mathematically equivalent to the no-signaling condition in Bell scenarios. Moreover, when we view a Bell scenario as a compatibility scenario the noncontextual polytope in that scenario coincides with the local polytope~\cite{amaral2017geometrical,choudhary2024lifting}. The behaviors described by some factorizable hidden-variable model as described by Eq.~\eqref{eq: noncontextual factorizable} become equivalent to what is known as a \emph{local} hidden variable model.

\subsubsection{Nonconvex set of local operations and shared randomness}

Various notions of free operations have been introduced to investigate nonclassicality within Bell scenarios, captured by violations of Bell inequalities. In particular, for the case of \emph{local operations and shared randomness}, the situation is rather confusing (as pointed out by Ref.~\cite{wolfe2020quantifying}) as different proposals have this same name. Let us start by defining $\mathrm{LOSR}^\circ$ as described in Ref.~\cite{gallego2017nonlocality}. For our purposes, it will suffice to discuss bipartite Bell scenarios only. Consider a scenario where Alice can perform $X$ measurements with outcomes $A$ and Bob can perform $Y$ measurements with outcomes $B$. The compatibility scenario associated with this Bell scenario is described by one where the set of measurements is $X \cup Y$, the set of maximal contexts are the pairs $\{x,y\}$ and finally, the set of all outcomes $\mathcal{O}^{\mathcal{M}}$ is such that all joint outcomes satisfy  $\mathcal{O}^{\{x,y\}} = A \times B$. The behaviors in this scenario are given by $\{p(ab|xy)\}_{a,b,x,y}$. 

Fig.~\ref{fig: LOSR^circ} depicts a transformation within the set $\mathrm{LOSR}^\circ$ acting on a bipartite Bell scenario. An initial common source of randomness, denoted by the preparation of some state $\lambda$ with probability $p(\lambda)$, is shared between Alice and Bob, each of which perform some set of experiments having classical inputs $x',y'$ (choices of measurements) and classical outputs $x,y$ (outcomes of their measurements). 

Each party has an associated side channel (e.g., a classical memory) that copies the classical inputs and outputs and store this information to be used later during a post-processing stage. Based upon the measurements $a$ at this intermediate stage, Alice then uses this information, and the information of $x'$ and $x$ previously stored, to choose some other post-processing measurements using some function  $g_{\mathrm{A}}(a|x',x)$, and obtain a final outcome $a'$, similarly for Bob. In this case, we describe these operations by some mapping $B = \{p(ab|xy)\} \mapsto \{p(a'b'|x'y')\} = B'$ where $B' = \mathcal{L}(B)$ for some $\mathcal{L} \in \mathrm{LOSR}^\circ$. Writing in full detail, these operations must have the form

\begin{align}    &p(a'b'|x'y')=\sum_{a,b,x,y}p(xy|x'y')p(ab|x,y)\times \label{eq: LOSR^circ}\\&\times  p(a'b'|g_{\mathrm{A}}(a|x,x'),g_{\mathrm{B}}(b|y,y')),\nonumber
\end{align}
where above both the pre-processing distribution  $p(xy|x'y')$ and the post-processing distribution $p(a'b'|g_{\mathrm{A}}(a|x,x'),g_{\mathrm{B}}(b|y,y'))$ are described by some local hidden-variable model. Also, each distribution has its own common source of classical randomness (see Fig.~\ref{fig: LOSR^circ}).

\begin{figure}
    \centering
    \begin{tikzpicture}
	\begin{pgfonlayer}{nodelayer}
		\node [style=none] (0) at (-1.25, 1.5) {};
		\node [style=none] (1) at (-2.25, 1.5) {};
		\node [style=none] (2) at (-2.25, 1) {};
		\node [style=none] (3) at (-0.75, 1) {};
		\node [style=none] (4) at (-1.75, 1.5) {};
		\node [style=none] (5) at (-1.75, 1) {};
		\node [style=none] (6) at (-1.75, 1.75) {};
		\node [style=none] (7) at (-2.25, 1.75) {};
		\node [style=none] (8) at (-1.25, 1.75) {};
		\node [style=none] (9) at (-1.75, 2.5) {};
		\node [style=none] (14) at (0.75, 1.5) {};
		\node [style=none] (15) at (1.75, 1.5) {};
		\node [style=none] (16) at (1.75, 1) {};
		\node [style=none] (17) at (0.25, 1) {};
		\node [style=none] (18) at (0.75, 1) {};
		\node [style=none] (19) at (1.25, 1.5) {};
		\node [style=none] (20) at (1.25, 1) {};
		\node [style=none] (21) at (1.25, 1.75) {};
		\node [style=none] (22) at (0.75, 1.75) {};
		\node [style=none] (23) at (1.75, 1.75) {};
		\node [style=none] (24) at (1.25, 2.5) {};
		\node [style=none] (25) at (1.25, 2) {$b$};
		\node [style=none] (30) at (-0.75, 0.25) {};
		\node [style=none] (31) at (0.25, 0.25) {};
		\node [style=none] (32) at (-0.25, -0.5) {};
		\node [style=none] (33) at (-0.25, 0) {$\Psi$};
		\node [style=none] (34) at (0, 0.25) {};
		\node [style=none] (35) at (-0.5, 0.25) {};
		\node [style=none] (36) at (-1.25, 1) {};
		\node [style=none] (37) at (-1.75, 2) {$a$};
		\node [style=none] (38) at (2.25, 2.5) {};
		\node [style=none] (39) at (1.75, 2.5) {};
		\node [style=none] (40) at (2.75, 2.5) {};
		\node [style=none] (41) at (2.25, 1.75) {};
		\node [style=none] (42) at (2.25, 2.25) {$b$};
		\node [style=none] (43) at (-2.75, 2.5) {};
		\node [style=none] (44) at (-3.25, 2.5) {};
		\node [style=none] (45) at (-2.25, 2.5) {};
		\node [style=none] (46) at (-2.75, 1.75) {};
		\node [style=none] (47) at (-2.75, 2.25) {$a$};
		\node [style=none] (48) at (2, 0) {};
		\node [style=none] (49) at (1.5, 0) {};
		\node [style=none] (50) at (2.5, 0) {};
		\node [style=none] (51) at (2, 0.75) {};
		\node [style=none] (52) at (2, 0.25) {$y$};
		\node [style=none] (53) at (-2.5, 0) {};
		\node [style=none] (54) at (-3, 0) {};
		\node [style=none] (55) at (-2, 0) {};
		\node [style=none] (56) at (-2.5, 0.75) {};
		\node [style=none] (57) at (-2.5, 0.25) {$x$};
		\node [style=none] (58) at (-0.5, -1.75) {};
		\node [style=none] (59) at (-1, -1.75) {};
		\node [style=none] (60) at (0.5, -1.75) {};
		\node [style=none] (61) at (-0.25, -2.5) {};
		\node [style=none] (62) at (-0.25, -2) {$\lambda$};
		\node [style=none] (63) at (-3, -0.5) {};
		\node [style=none] (64) at (-2, -0.5) {};
		\node [style=none] (65) at (-3, -1) {};
		\node [style=none] (66) at (-2, -1) {};
		\node [style=none] (67) at (-1.5, -1) {};
		\node [style=none] (68) at (2.5, -1) {};
		\node [style=none] (69) at (1, -1) {};
		\node [style=none] (70) at (2.5, -0.5) {};
		\node [style=none] (71) at (2.5, -0.5) {};
		\node [style=none] (72) at (1.5, -0.5) {};
		\node [style=none] (73) at (0, -1.75) {};
		\node [style=none] (74) at (2, -1) {};
		\node [style=none] (75) at (1.5, -1) {};
		\node [style=none] (76) at (-2.5, -1) {};
		\node [style=none] (77) at (-2.5, -1.75) {};
		\node [style=none] (78) at (-3, -1.75) {};
		\node [style=none] (79) at (-2, -1.75) {};
		\node [style=none] (80) at (-2.5, -2.5) {};
		\node [style=none] (81) at (-2.5, -2) {$x'$};
		\node [style=none] (82) at (2, -1.75) {};
		\node [style=none] (83) at (1.5, -1.75) {};
		\node [style=none] (84) at (2.5, -1.75) {};
		\node [style=none] (85) at (2, -2.5) {};
		\node [style=none] (86) at (2, -2) {$y'$};
		\node [style=none] (87) at (1.75, 4.75) {};
		\node [style=none] (88) at (1.25, 4.75) {};
		\node [style=none] (89) at (2.25, 4.75) {};
		\node [style=none] (90) at (1.75, 5.5) {};
		\node [style=none] (91) at (1.75, 5) {$b'$};
		\node [style=none] (92) at (-2.5, 4.75) {};
		\node [style=none] (93) at (-3, 4.75) {};
		\node [style=none] (94) at (-2, 4.75) {};
		\node [style=none] (95) at (-2.5, 5.5) {};
		\node [style=none] (96) at (-2.5, 5) {$a'$};
		\node [style=none] (97) at (-0.5, 2.75) {};
		\node [style=none] (98) at (-1, 2.75) {};
		\node [style=none] (99) at (0.5, 2.75) {};
		\node [style=none] (100) at (-0.25, 2) {};
		\node [style=none] (101) at (-0.25, 2.5) {$\phi$};
		\node [style=none] (102) at (-3, 4.25) {};
		\node [style=none] (103) at (-2, 4.25) {};
		\node [style=none] (104) at (-3, 3.75) {};
		\node [style=none] (105) at (-2, 3.75) {};
		\node [style=none] (106) at (-1.5, 3.75) {};
		\node [style=none] (107) at (2.25, 3.75) {};
		\node [style=none] (108) at (0.75, 3.75) {};
		\node [style=none] (109) at (2.25, 4.25) {};
		\node [style=none] (110) at (2.25, 4.25) {};
		\node [style=none] (111) at (1.25, 4.25) {};
		\node [style=none] (112) at (0, 2.75) {};
		\node [style=none] (113) at (1.75, 3.75) {};
		\node [style=none] (114) at (1.25, 3.75) {};
		\node [style=none] (115) at (-2.5, 3.75) {};
		\node [style=none] (116) at (-1.75, 0.75) {};
		\node [style=none] (117) at (-2.25, 0.75) {};
		\node [style=none] (118) at (-1.25, 0.75) {};
		\node [style=none] (119) at (-1.75, 0) {};
		\node [style=none] (120) at (-1.75, 0.5) {$x$};
		\node [style=none] (126) at (1.25, 0.75) {};
		\node [style=none] (127) at (0.75, 0.75) {};
		\node [style=none] (128) at (1.75, 0.75) {};
		\node [style=none] (129) at (1.25, 0) {};
		\node [style=none] (130) at (1.25, 0.5) {$y$};
		\node [style=none] (132) at (-2.5, -0.5) {};
		\node [style=none] (133) at (2, -0.5) {};
		\node [style=none] (134) at (-2.5, 4.25) {};
		\node [style=none] (135) at (1.75, 4.25) {};
		\node [style=none] (136) at (1.25, 3.5) {};
		\node [style=none] (137) at (1.25, 3) {};
		\node [style=none] (138) at (4, 3.5) {};
		\node [style=none] (139) at (4, 3) {};
		\node [style=none] (140) at (1.75, 3.5) {};
		\node [style=none] (141) at (2.25, 3) {};
		\node [style=none] (142) at (3, 3) {};
		\node [style=none] (143) at (3.5, 3) {};
		\node [style=none] (144) at (-4.5, 3.5) {};
		\node [style=none] (145) at (-4.5, 3) {};
		\node [style=none] (146) at (-2, 3.5) {};
		\node [style=none] (147) at (-2, 3) {};
		\node [style=none] (148) at (-2.5, 3.5) {};
		\node [style=none] (149) at (-3.25, 3.25) {$g_{\mathrm{A}}$};
		\node [style=none] (150) at (2.75, 3.25) {$g_{\mathrm{B}}$};
		\node [style=copy] (151) at (2, -1.25) {};
		\node [style=copy] (152) at (-2.5, -1.25) {};
		\node [style=copy] (153) at (-2.5, -0.25) {};
		\node [style=copy] (154) at (2, -0.25) {};
		\node [style=none] (155) at (3, 0.5) {};
		\node [style=none] (156) at (3.5, 0.5) {};
		\node [style=none] (157) at (-3.5, 3) {};
		\node [style=none] (158) at (-4, 3) {};
		\node [style=none] (159) at (-3.5, 0.5) {};
		\node [style=none] (160) at (-4, 0.5) {};
		\node [style=none] (161) at (-2.75, 3) {};
	\end{pgfonlayer}
	\begin{pgfonlayer}{edgelayer}
		\draw [style=alice] (1.center)
			 to (0.center)
			 to (3.center)
			 to (2.center)
			 to cycle;
		\draw (4.center) to (6.center);
		\draw (8.center) to (7.center);
		\draw (7.center) to (9.center);
		\draw (9.center) to (8.center);
		\draw [style=alice] (15.center)
			 to (14.center)
			 to (17.center)
			 to (16.center)
			 to cycle;
		\draw (19.center) to (21.center);
		\draw (23.center) to (22.center);
		\draw (22.center) to (24.center);
		\draw (24.center) to (23.center);
		\draw [style=alice] (30.center)
			 to (32.center)
			 to (31.center)
			 to cycle;
		\draw [style=alice, in=-90, out=90] (34.center) to (18.center);
		\draw [style=alice, in=-90, out=90, looseness=0.75] (35.center) to (36.center);
		\draw (39.center) to (40.center);
		\draw (39.center) to (41.center);
		\draw (41.center) to (40.center);
		\draw (44.center) to (45.center);
		\draw (44.center) to (46.center);
		\draw (46.center) to (45.center);
		\draw (50.center) to (49.center);
		\draw (49.center) to (51.center);
		\draw (51.center) to (50.center);
		\draw (55.center) to (54.center);
		\draw (54.center) to (56.center);
		\draw (56.center) to (55.center);
		\draw (59.center) to (60.center);
		\draw (59.center) to (61.center);
		\draw (61.center) to (60.center);
		\draw (63.center) to (64.center);
		\draw (65.center) to (63.center);
		\draw (67.center) to (64.center);
		\draw (66.center) to (67.center);
		\draw (65.center) to (66.center);
		\draw (68.center) to (69.center);
		\draw (70.center) to (68.center);
		\draw (72.center) to (69.center);
		\draw (71.center) to (72.center);
		\draw (70.center) to (71.center);
		\draw [in=-90, out=90, looseness=0.75] (58.center) to (66.center);
		\draw [in=-90, out=90, looseness=0.75] (73.center) to (75.center);
		\draw (78.center) to (79.center);
		\draw (78.center) to (80.center);
		\draw (80.center) to (79.center);
		\draw (83.center) to (84.center);
		\draw (83.center) to (85.center);
		\draw (85.center) to (84.center);
		\draw (89.center) to (88.center);
		\draw (88.center) to (90.center);
		\draw (90.center) to (89.center);
		\draw (94.center) to (93.center);
		\draw (93.center) to (95.center);
		\draw (95.center) to (94.center);
		\draw (98.center) to (99.center);
		\draw (98.center) to (100.center);
		\draw (100.center) to (99.center);
		\draw (102.center) to (103.center);
		\draw (104.center) to (102.center);
		\draw (106.center) to (103.center);
		\draw (105.center) to (106.center);
		\draw (104.center) to (105.center);
		\draw (107.center) to (108.center);
		\draw (109.center) to (107.center);
		\draw (111.center) to (108.center);
		\draw (110.center) to (111.center);
		\draw (109.center) to (110.center);
		\draw [in=-90, out=90] (97.center) to (105.center);
		\draw [in=-90, out=90] (112.center) to (114.center);
		\draw (117.center) to (118.center);
		\draw (117.center) to (119.center);
		\draw (119.center) to (118.center);
		\draw (127.center) to (128.center);
		\draw (127.center) to (129.center);
		\draw (129.center) to (128.center);
		\draw (77.center) to (76.center);
		\draw (82.center) to (74.center);
		\draw (133.center) to (48.center);
		\draw (132.center) to (53.center);
		\draw (134.center) to (92.center);
		\draw (135.center) to (87.center);
		\draw (137.center) to (136.center);
		\draw (136.center) to (138.center);
		\draw (138.center) to (139.center);
		\draw (139.center) to (137.center);
		\draw (140.center) to (113.center);
		\draw (38.center) to (141.center);
		\draw (145.center) to (144.center);
		\draw (144.center) to (146.center);
		\draw (146.center) to (147.center);
		\draw (147.center) to (145.center);
		\draw (148.center) to (115.center);
		\draw (155.center) to (142.center);
		\draw (156.center) to (143.center);
		\draw [in=-90, out=0] (154) to (155.center);
		\draw [in=0, out=-90, looseness=1.25] (156.center) to (151);
		\draw (160.center) to (158.center);
		\draw (159.center) to (157.center);
		\draw [in=-180, out=-90] (159.center) to (153);
		\draw [in=180, out=-90, looseness=1.25] (160.center) to (152);
		\draw (43.center) to (161.center);
		\draw (116.center) to (5.center);
		\draw (126.center) to (20.center);
	\end{pgfonlayer}
\end{tikzpicture}
    \caption{Example of a local operation with shared randomness from the set $\mathrm{LOSR}^\circ$. For simplicity we have omitted the terms $\sum_\lambda p(\lambda), \sum_\phi p(\phi)$. Note that each pre- and post-processing stages of operations have independent sources of classical randomness. Furthermore, side channels only forward information from pre- to post-processing locally.}
    \label{fig: LOSR^circ}
\end{figure}
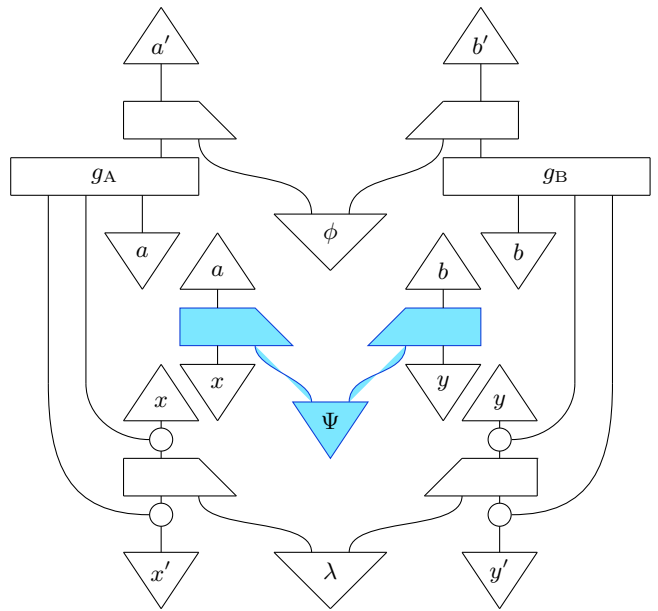

It is straightforward to generalize this operational description to multipartite Bell scenarios. Ref.~\cite{gallego2017nonlocality} showed that these operations are free, in the sense that if a behavior has a locally causal explanation given by some local hidden-variable model, once a transformation $\mathcal{L} \in \mathrm{LOSR}^\circ$ takes place, the new behavior also has some locally causal explanation. What we are interested here is in making the following simple remark
\begin{equation}
    \mathrm{LOSR}^\circ \subsetneq \mathrm{NCW}|_{\text{Bell}},
\end{equation}
Where $\mathrm{NCW}|_{\text{Bell}}$ is the subset of all noncontextual wiring transformations of type $\mathrm{NCW}(\Upsilon_1^{\text{Bell}} \to \Upsilon_2^{\text{Bell}})$ where both $\Upsilon_1^{\text{Bell}},\Upsilon_2^{\text{Bell}}$ are compatibility scenarios mathematically isomorphic to some Bell scenario. In words, when restricted to Bell scenarios, the noncontextual wirings contain, but are not equivalent to the local operations and shared randomness described by transformations in $\mathrm{LOSR}^\circ$, and characterized operationally as in Fig.~\ref{fig: LOSR^circ}.

To see why this is the case one compares the operation presented in Fig.~\ref{fig: LOSR^circ} with the noncontextual wiring presented in Fig.~\ref{fig: noncontextual wirings}. During the post-processing stage, noncontextual wirings allow for choosing a specific input \emph{context} provided we have information of a past \emph{context} considered. In particular, every transformation in $\mathrm{LOSR}^\circ$ can be described in terms of noncontextual wirings in such scenarios. However, a generic $\mathrm{NCW}|_{\text{Bell}}$ operation also allows both parties to have information about both classical inputs/outputs of the pre-processing box to choose a post-processing. Operationally, this implies that general NCW operations allow for a classical side channel sending the information of inputs $x'$ and $x$ made by Alice during a pre-processing stage on her side to Bob's choices of post-processing, and vice-versa. This corresponds to a choice of a generic function $g(\cdot | \mathbf{r},\beta)$ that has inputs of the contexts used in the pre-processing stage (in this case, $\beta = \{x',y'\}$) and their joint outcomes (in this case, $\mathbf{r} = (x,y)$).

This is a subtle difference. To be concrete, if we consider Eq.~\eqref{eq: LOSR^circ}, an operation in $\mathrm{NCW}|_{\text{Bell}}$ also allow that each function $g_{\mathrm{A}}$ takes into consideration $y$ and $y'$, and vice-versa. In this case, we might have that Alice chooses her measurements using  $g_{\mathrm{A}}(a|x,y,x',y')$ and the same for Bob $g_{\mathrm{B}}(b|x,y,x',y')$. Refs.~\cite{amaral2018noncontextual,amaral2019resource} have overlooked this subtle difference, which became clear due to the discussion provided by Ref.~\cite{wolfe2020quantifying}. Notably, this remark was made by Karvonen in Ref.~\cite{karvonen2021neither}, who wrote ``At the level of transformations, nonlocality is no longer a special case of contextuality. (...) This is because our wirings are slightly too general as they allow the $i$-th party to wire some of their measurements to measurements belonging to other parties.''  The only remark to be made here is that $g$ \emph{cannot} be completely general (as opposed to $f$), as this would allow some non-local (or contextual) correlations to be created. When applied to Bell scenarios, $g$ must also factorize its influence, i.e., $g = g_{\mathrm{A}} \cup g_{\mathrm{B}}$, as in Fig.~\ref{fig: LOSR^circ}. In more general cases, we simply require that $g(\mathbf{s}|\mathbf{r},\beta) = \cup_i g_i(s_i|\mathbf{r},\beta))$ as described in Def.~\ref{def: noncontextual wirings}. The main result of this subsection (and the following one) is then making Karvonen's remark concrete in the context of noncontextual wirings, a \emph{strict} subset of the free operations Karvonen considered in Ref.~\cite{karvonen2021neither}. 

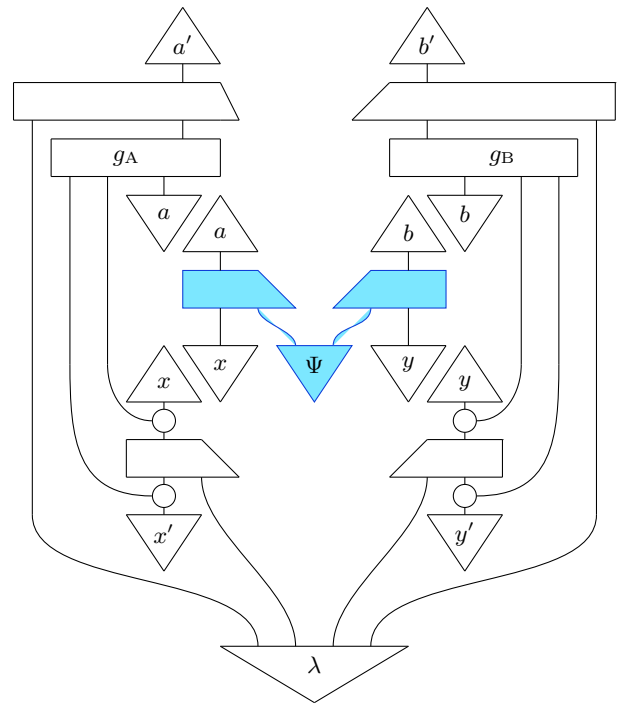
\begin{figure}
    \centering
    \begin{tikzpicture}
	\begin{pgfonlayer}{nodelayer}
		\node [style=none] (0) at (-0.75, 2.25) {};
		\node [style=none] (1) at (-1.75, 2.25) {};
		\node [style=none] (2) at (-1.75, 1.75) {};
		\node [style=none] (3) at (-0.25, 1.75) {};
		\node [style=none] (4) at (-1.25, 2.25) {};
		\node [style=none] (5) at (-1.25, 1.75) {};
		\node [style=none] (6) at (-1.25, 2.5) {};
		\node [style=none] (7) at (-1.75, 2.5) {};
		\node [style=none] (8) at (-0.75, 2.5) {};
		\node [style=none] (9) at (-1.25, 3.25) {};
		\node [style=none] (14) at (0.75, 2.25) {};
		\node [style=none] (15) at (1.75, 2.25) {};
		\node [style=none] (16) at (1.75, 1.75) {};
		\node [style=none] (17) at (0.25, 1.75) {};
		\node [style=none] (18) at (0.75, 1.75) {};
		\node [style=none] (19) at (1.25, 2.25) {};
		\node [style=none] (20) at (1.25, 1.75) {};
		\node [style=none] (21) at (1.25, 2.5) {};
		\node [style=none] (22) at (0.75, 2.5) {};
		\node [style=none] (23) at (1.75, 2.5) {};
		\node [style=none] (24) at (1.25, 3.25) {};
		\node [style=none] (25) at (1.25, 2.75) {$b$};
		\node [style=none] (30) at (-0.5, 1.25) {};
		\node [style=none] (31) at (0.5, 1.25) {};
		\node [style=none] (32) at (0, 0.5) {};
		\node [style=none] (33) at (0, 1) {$\Psi$};
		\node [style=none] (34) at (0.25, 1.25) {};
		\node [style=none] (35) at (-0.25, 1.25) {};
		\node [style=none] (36) at (-0.75, 1.75) {};
		\node [style=none] (37) at (-1.25, 2.75) {$a$};
		\node [style=none] (38) at (2, 3.25) {};
		\node [style=none] (39) at (1.5, 3.25) {};
		\node [style=none] (40) at (2.5, 3.25) {};
		\node [style=none] (41) at (2, 2.5) {};
		\node [style=none] (42) at (2, 3) {$b$};
		\node [style=none] (43) at (-2, 3.25) {};
		\node [style=none] (44) at (-2.5, 3.25) {};
		\node [style=none] (45) at (-1.5, 3.25) {};
		\node [style=none] (46) at (-2, 2.5) {};
		\node [style=none] (47) at (-2, 3) {$a$};
		\node [style=none] (48) at (2, 0.5) {};
		\node [style=none] (49) at (1.5, 0.5) {};
		\node [style=none] (50) at (2.5, 0.5) {};
		\node [style=none] (51) at (2, 1.25) {};
		\node [style=none] (52) at (2, 0.75) {$y$};
		\node [style=none] (53) at (-2, 0.5) {};
		\node [style=none] (54) at (-2.5, 0.5) {};
		\node [style=none] (55) at (-1.5, 0.5) {};
		\node [style=none] (56) at (-2, 1.25) {};
		\node [style=none] (57) at (-2, 0.75) {$x$};
		\node [style=none] (58) at (-0.25, -2.75) {};
		\node [style=none] (59) at (-1.25, -2.75) {};
		\node [style=none] (60) at (1.25, -2.75) {};
		\node [style=none] (61) at (0, -3.5) {};
		\node [style=none] (62) at (0, -3) {$\lambda$};
		\node [style=none] (63) at (-2.5, 0) {};
		\node [style=none] (64) at (-1.5, 0) {};
		\node [style=none] (65) at (-2.5, -0.5) {};
		\node [style=none] (66) at (-1.5, -0.5) {};
		\node [style=none] (67) at (-1, -0.5) {};
		\node [style=none] (68) at (2.5, -0.5) {};
		\node [style=none] (69) at (1, -0.5) {};
		\node [style=none] (70) at (2.5, 0) {};
		\node [style=none] (71) at (2.5, 0) {};
		\node [style=none] (72) at (1.5, 0) {};
		\node [style=none] (73) at (0.25, -2.75) {};
		\node [style=none] (74) at (2, -0.5) {};
		\node [style=none] (75) at (1.5, -0.5) {};
		\node [style=none] (76) at (-2, -0.5) {};
		\node [style=none] (77) at (-2, -1) {};
		\node [style=none] (78) at (-2.5, -1) {};
		\node [style=none] (79) at (-1.5, -1) {};
		\node [style=none] (80) at (-2, -1.75) {};
		\node [style=none] (81) at (-2, -1.25) {$x'$};
		\node [style=none] (82) at (2, -1) {};
		\node [style=none] (83) at (1.5, -1) {};
		\node [style=none] (84) at (2.5, -1) {};
		\node [style=none] (85) at (2, -1.75) {};
		\node [style=none] (86) at (2, -1.25) {$y'$};
		\node [style=none] (87) at (1.5, 5) {};
		\node [style=none] (88) at (1, 5) {};
		\node [style=none] (89) at (2, 5) {};
		\node [style=none] (90) at (1.5, 5.75) {};
		\node [style=none] (91) at (1.5, 5.25) {$b'$};
		\node [style=none] (92) at (-1.75, 5) {};
		\node [style=none] (93) at (-2.25, 5) {};
		\node [style=none] (94) at (-1.25, 5) {};
		\node [style=none] (95) at (-1.75, 5.75) {};
		\node [style=none] (96) at (-1.75, 5.25) {$a'$};
		\node [style=none] (97) at (-4, 4.75) {};
		\node [style=none] (98) at (-1.25, 4.75) {};
		\node [style=none] (99) at (-4, 4.25) {};
		\node [style=none] (100) at (-1.25, 4.25) {};
		\node [style=none] (101) at (-1, 4.25) {};
		\node [style=none] (102) at (4, 4.25) {};
		\node [style=none] (103) at (0.5, 4.25) {};
		\node [style=none] (104) at (4, 4.75) {};
		\node [style=none] (105) at (4, 4.75) {};
		\node [style=none] (106) at (1, 4.75) {};
		\node [style=none] (107) at (1.5, 4.25) {};
		\node [style=none] (108) at (1, 4.25) {};
		\node [style=none] (109) at (-1.75, 4.25) {};
		\node [style=none] (110) at (-1.25, 1.25) {};
		\node [style=none] (111) at (-1.75, 1.25) {};
		\node [style=none] (112) at (-0.75, 1.25) {};
		\node [style=none] (113) at (-1.25, 0.5) {};
		\node [style=none] (114) at (-1.25, 1) {$x$};
		\node [style=none] (120) at (1.25, 1.25) {};
		\node [style=none] (121) at (0.75, 1.25) {};
		\node [style=none] (122) at (1.75, 1.25) {};
		\node [style=none] (123) at (1.25, 0.5) {};
		\node [style=none] (124) at (1.25, 1) {$y$};
		\node [style=none] (126) at (-2, 0) {};
		\node [style=none] (127) at (2, 0) {};
		\node [style=none] (128) at (-1.75, 4.75) {};
		\node [style=none] (129) at (1.5, 4.75) {};
		\node [style=none] (130) at (1, 4) {};
		\node [style=none] (131) at (1, 3.5) {};
		\node [style=none] (132) at (3.5, 4) {};
		\node [style=none] (133) at (3.5, 3.5) {};
		\node [style=none] (134) at (1.5, 4) {};
		\node [style=none] (135) at (2, 3.5) {};
		\node [style=none] (136) at (2.75, 3.5) {};
		\node [style=none] (137) at (3.25, 3.5) {};
		\node [style=none] (138) at (-3.5, 4) {};
		\node [style=none] (139) at (-3.5, 3.5) {};
		\node [style=none] (140) at (-1.25, 4) {};
		\node [style=none] (141) at (-1.25, 3.5) {};
		\node [style=none] (142) at (-1.75, 4) {};
		\node [style=none] (143) at (-2.5, 3.75) {$g_{\mathrm{A}}$};
		\node [style=none] (144) at (2.5, 3.75) {$g_{\mathrm{B}}$};
		\node [style=copy] (145) at (2, -0.75) {};
		\node [style=copy] (146) at (-2, -0.75) {};
		\node [style=copy] (147) at (-2, 0.25) {};
		\node [style=copy] (148) at (2, 0.25) {};
		\node [style=none] (149) at (2.75, 1) {};
		\node [style=none] (150) at (3.25, 1) {};
		\node [style=none] (151) at (-2.75, 3.5) {};
		\node [style=none] (152) at (-3.25, 3.5) {};
		\node [style=none] (153) at (-2.75, 1) {};
		\node [style=none] (154) at (-3.25, 1) {};
		\node [style=none] (155) at (-2, 3.5) {};
		\node [style=none] (156) at (0.75, -2.75) {};
		\node [style=none] (157) at (-0.75, -2.75) {};
		\node [style=none] (158) at (-3.75, 4.25) {};
		\node [style=none] (159) at (3.75, 4.25) {};
		\node [style=none] (160) at (-3.75, -1) {};
		\node [style=none] (161) at (3.75, -1) {};
	\end{pgfonlayer}
	\begin{pgfonlayer}{edgelayer}
		\draw [style=alice] (3.center)
			 to (2.center)
			 to (1.center)
			 to (0.center)
			 to cycle;
		\draw (4.center) to (6.center);
		\draw (8.center) to (7.center);
		\draw (7.center) to (9.center);
		\draw (9.center) to (8.center);
		\draw [style=alice] (16.center)
			 to (15.center)
			 to (14.center)
			 to (17.center)
			 to cycle;
		\draw (19.center) to (21.center);
		\draw (23.center) to (22.center);
		\draw (22.center) to (24.center);
		\draw (24.center) to (23.center);
		\draw [style=alice] (32.center)
			 to (31.center)
			 to (30.center)
			 to cycle;
		\draw [style=alice, in=-90, out=90] (34.center) to (18.center);
		\draw [style=alice, in=-90, out=90, looseness=0.75] (35.center) to (36.center);
		\draw (39.center) to (40.center);
		\draw (39.center) to (41.center);
		\draw (41.center) to (40.center);
		\draw (44.center) to (45.center);
		\draw (44.center) to (46.center);
		\draw (46.center) to (45.center);
		\draw (50.center) to (49.center);
		\draw (49.center) to (51.center);
		\draw (51.center) to (50.center);
		\draw (55.center) to (54.center);
		\draw (54.center) to (56.center);
		\draw (56.center) to (55.center);
		\draw (59.center) to (60.center);
		\draw (59.center) to (61.center);
		\draw (61.center) to (60.center);
		\draw (63.center) to (64.center);
		\draw (65.center) to (63.center);
		\draw (67.center) to (64.center);
		\draw (66.center) to (67.center);
		\draw (65.center) to (66.center);
		\draw (68.center) to (69.center);
		\draw (70.center) to (68.center);
		\draw (72.center) to (69.center);
		\draw (71.center) to (72.center);
		\draw (70.center) to (71.center);
		\draw [in=-90, out=90, looseness=0.75] (58.center) to (66.center);
		\draw [in=-90, out=90, looseness=0.75] (73.center) to (75.center);
		\draw (78.center) to (79.center);
		\draw (78.center) to (80.center);
		\draw (80.center) to (79.center);
		\draw (83.center) to (84.center);
		\draw (83.center) to (85.center);
		\draw (85.center) to (84.center);
		\draw (89.center) to (88.center);
		\draw (88.center) to (90.center);
		\draw (90.center) to (89.center);
		\draw (94.center) to (93.center);
		\draw (93.center) to (95.center);
		\draw (95.center) to (94.center);
		\draw (97.center) to (98.center);
		\draw (99.center) to (97.center);
		\draw (101.center) to (98.center);
		\draw (100.center) to (101.center);
		\draw (99.center) to (100.center);
		\draw (102.center) to (103.center);
		\draw (104.center) to (102.center);
		\draw (106.center) to (103.center);
		\draw (105.center) to (106.center);
		\draw (104.center) to (105.center);
		\draw (111.center) to (112.center);
		\draw (111.center) to (113.center);
		\draw (113.center) to (112.center);
		\draw (121.center) to (122.center);
		\draw (121.center) to (123.center);
		\draw (123.center) to (122.center);
		\draw (77.center) to (76.center);
		\draw (82.center) to (74.center);
		\draw (127.center) to (48.center);
		\draw (126.center) to (53.center);
		\draw (128.center) to (92.center);
		\draw (129.center) to (87.center);
		\draw (131.center) to (130.center);
		\draw (130.center) to (132.center);
		\draw (132.center) to (133.center);
		\draw (133.center) to (131.center);
		\draw (134.center) to (107.center);
		\draw (38.center) to (135.center);
		\draw (139.center) to (138.center);
		\draw (138.center) to (140.center);
		\draw (140.center) to (141.center);
		\draw (141.center) to (139.center);
		\draw (142.center) to (109.center);
		\draw (149.center) to (136.center);
		\draw (150.center) to (137.center);
		\draw [in=-90, out=0] (148) to (149.center);
		\draw [in=0, out=-90, looseness=1.25] (150.center) to (145);
		\draw (154.center) to (152.center);
		\draw (153.center) to (151.center);
		\draw [in=-180, out=-90] (153.center) to (147);
		\draw [in=180, out=-90, looseness=1.25] (154.center) to (146);
		\draw (43.center) to (155.center);
		\draw (160.center) to (158.center);
		\draw (161.center) to (159.center);
		\draw [in=90, out=-90, looseness=0.75] (161.center) to (156.center);
		\draw [in=90, out=-90, looseness=0.75] (160.center) to (157.center);
		\draw (120.center) to (20.center);
		\draw (110.center) to (5.center);
	\end{pgfonlayer}
\end{tikzpicture}
    \caption{Example of a local operation with shared randomness from the set $\mathrm{LOSR}^\bullet$. For simplicity, we have omitted the terms $\sum_\lambda p(\lambda)$. Note that classical feedforward and adaptivity are only allowed locally and that a single common source of randomness is associated with both the pre- and post-processing boxes. As argued in Ref.~\cite{wolfe2020quantifying}, this is the decisive operational requirement that allows $\mathrm{LOSR}^\bullet$ to be a convex set.}
    \label{fig: LOSR^bullet}
\end{figure}

Naturally, the two notions differ. Noncontextuality scenarios do not base their operational constraints on some notion of spatial separability between different parties. Operations characterized by some notion of local operation and shared randomness \emph{intentionally} use the operational structure of Bell scenarios to characterize the possible transformations. $\mathrm{LOSR}^\circ$ (and later what we will denote $\mathrm{LOSR}^\bullet$) restrict parties to only process classical information that is locally available, while allowing initial sources of shared randomness. A KS-noncontextuality scenario that is mathematically isomorphic to some Bell scenario has no operational motivation to restrict classical feedforward between measurements that can be jointly performed, since those are not necessarily measured by space-like separated parties, and can be, for example, measured sequentially, as was done, e.g., in Ref.~\cite{wang2022significant}.

\subsubsection{Convex set of local operations and shared randomness}

For completeness, we include a discussion from Ref.~\cite[App. A]{wolfe2020quantifying}.  There, the authors show that the set of local operations with shared randomness of a fixed type  $\mathrm{LOSR}^\circ(\Upsilon^{\mathrm{Bell}}_{\mathrm{in}} \to \Upsilon^{\mathrm{Bell}}_{\mathrm{out}})$ is a set of free operations that is \emph{not convex}~\footnote{They do not prove, but we believe that even if one relaxes the constraint of type-dependence the set $\mathrm{LOSR}^\circ$ might be non-convex as well.}. This is shown by creating a convex combination of free operations such that the final one has a perfect correlation between $x$ and $b'$ that cannot be mediated by $y$, because $y$ in such an example ends up being deterministic (always the same). 

Ref.~\cite{wolfe2020quantifying} considers a different set of operations, that is mathematically described as the convex hull of $\mathrm{LOSR}^\circ$, that we shall denote $\mathrm{LOSR}^\bullet$. This set was introduced and investigated in Refs.~\cite{deVicente2014LOSR,geller2014quantifying}. This convex hull is not only abstractly imposed at a mathematical level, but it is the result of a novel link allowed by an operational description of the free transformations. They allow the initial source of randomness to be both the source of randomness for the pre-processing \emph{and} for the post-processing operations, as shown in Fig.~\ref{fig: LOSR^bullet}. Because of that, it follows immediately that 
\begin{equation}\label{eq: LOSR different than NCW}
    \mathrm{LOSR}^{\bullet} \neq \mathrm{NCW}|_{\text{Bell}}.
\end{equation}
These two sets of free operations are different.  The above follows simply from the same remark made before that in a generic noncontextual wiring both parties can have information about contexts $\{x',y'\}$ considered and joint outcomes $\{x,y\}$ of the pre-processing behaviors. 

Ref.~\cite{wolfe2020quantifying} credits, in our view correctly, the closure of $\mathrm{LOSR}^\bullet$ under convex combinations to this novel perspective on the common source of randomness of pre- and post-processings. This suggests that $\mathrm{NCW}|_{\text{Bell}}$, or even any set $\mathrm{NCW}(\Upsilon_{\mathrm{in}} \to \Upsilon_{\mathrm{out}})$, could be non-convex sets of free operations. In the next section, we investigate this aspect of the geometry of the set of free operations. 

\subsection{Convexity of the set of noncontextual wirings}\label{sec: convexity result}

As mentioned previously, an important technical aspect of a resource theory is the convexity of the chosen set of free operations. Let us begin by considering a general free operation in  $\mathrm{NCW}(\Upsilon_{\mathrm{in}} \to \Upsilon_{\mathrm{out}})$, as given by Eq.~\eqref{eq: Noncontextual wiring behavior}. We are fixing the input/output compatibility scenarios, that define the type of transformation. We are interested in investigating if this set is convex. Since we are fixing the type of a free operation, we fix the initial behaviors to be represented as  $B_{\mathrm{in}}= \{p_\gamma(\mathbf{s})\}$ from $\Upsilon_{\mathrm{in}}$ and the final behaviors to be represented as $B_{\mathrm{out}}=\{p_\beta(\mathbf{t})\}$ from $\Upsilon_{\mathrm{out}}$. 

Formally, we want to show that, for all $\alpha \in [0,1]$, and all $\mathcal{W}_1,\mathcal{W}_2 \in \mathrm{NCW}(\Upsilon_{\mathrm{in}}\to \Upsilon_{\mathrm{out}})$, there exists some wiring $\mathcal{W} \in \mathrm{NCW}(\Upsilon_{\mathrm{in}}\to \Upsilon_{\mathrm{out}})$ such that 
\begin{equation*}
    \mathcal{W}(B) = \alpha \mathcal{W}_1(B)+(1-\alpha)\mathcal{W}_2(B). 
\end{equation*}
We now describe the new noncontextual wiring $\mathcal{W}$. Consider that $\{p^{(1)}_\beta(\mathbf{r})\}$ and $\{p^{(2)}_\beta(\mathbf{r})\}$ are the two pre-processing behaviors for $\mathcal{W}_1$ and $\mathcal{W}_2$, respectively. Since they are both noncontextual, they will be given by

\begin{equation*}
    p_\beta^{(\Lambda)}(\mathbf{r}) := \sum_\lambda p^{(\Lambda)}(\lambda)\prod_{\beta_i \in \beta}p_{\beta_i}^{(\Lambda)}(\mathbf{r}|\lambda),
\end{equation*}
where $\Lambda \in \{1,2\}$ labels the respective wirings. We can describe this $\Lambda$ as a classical variable, and let $p(\Lambda)$ be a distribution over $\Lambda$ such that $p(1)=\alpha$ and $p(2)=1-\alpha$. Let us now construct a new behavior $\{\{p_\beta(\mathbf{r}')\}_{\mathbf{r}}\}_\beta$ such that $\mathbf{r}' \in \mathcal{O}^\beta \times \Lambda$. This behavior is given by, 
\begin{equation*}
    p_\beta(\mathbf{r}')= \sum_\Lambda p(\Lambda)\sum_\lambda p(\lambda|\Lambda)\prod_{\beta_i \in \beta}p_{\beta_i}^{(\Lambda)}(\mathbf{r}|\lambda)
\end{equation*}
where $p(\lambda|1) \equiv p^{(1)}(\lambda)$ the source of randomness associated with the noncontextual behavior $\{p_\beta^{(1)}(\mathbf{r})\}$, and similarly for $p(\lambda|2) \equiv p^{(2)}(\lambda)$. The behavior $p_\beta(\mathbf{r}')$ can be viewed as a new `larger' behavior $\{p_\beta^{(\Lambda)}(\mathbf{r})\}_{\Lambda \in \{1,2\}}$. The specific choice of state $\Lambda$ is passed forward to the outcomes of the new behavior (as a flag), via $\mathbf{r}' = (\mathbf{r},\Lambda)$. This flag is used to make the controlled choice between behavior $p_\beta^{(1)}(\mathbf{r})$ or $p_{\beta}^{(2)}(\mathbf{r})$. A similar thing will happen to the post-processing behavior. As shown in Ref.~\cite{abramsky2019comonadic}, the controlled choice of two noncontextual behaviors is again a noncontextual behavior. 

The function $f$ takes $\mathbf{r}'$ as input and is then defined as $f(\mathbf{r}') = f(\mathbf{r},\Lambda) := f^{(\Lambda)}(\mathbf{r})$. Note that, because information about $\Lambda$ is passed as a flag to the outcomes, this choice is possible. Moreover, we will make a similar choice for the function $g$ where we define $g(\mathbf{s}|\mathbf{r}',\beta) = g(\mathbf{s}|\mathbf{r},\Lambda,\beta) := g^{(\Lambda)}(\mathbf{s}|\mathbf{r},\beta)$. This is an important aspect: the function $f$ allows us to preserve the type of the scenario $\Upsilon_{\mathrm{in}}$. This is also why we can transfer information of $\Lambda$ to the post-processing stage until $g$, because $\mathbf{r}'$ carries information of $\Lambda$.

Because of that, we can define the post-processing behavior of $\mathcal{W}$ to be associated with either that of $\mathcal{W}_1$ or that of $\mathcal{W}_2$, depending on the flag value $\Lambda$ of the copied outcome $\mathbf{r}'$,
\begin{equation*}
    p_{\delta}(\mathbf{t}|\mathbf{s},\mathbf{r}',\beta) = \sum_\phi p(\phi|\Lambda)\prod_{\delta_i \in \delta=g(\mathbf{s}|\mathbf{r}',\beta)}p^{(\Lambda)}_{\delta_i}(t_i|\phi),  
\end{equation*}
i.e., $p_{\delta=g(\mathbf{s}|\mathbf{r}',\beta)}(\mathbf{t}) := p^{(\Lambda)}_{\delta = g(\mathbf{s}|\mathbf{r},\beta)}(\mathbf{t})$. Above, $p(\phi|\Lambda) = p^{(\Lambda)}(\phi)$ where $p^{(1)}(\phi)$ describes the randomness associated to the post-processing from $\mathcal{W}_1$, and similarly for $p^{(2)}(\phi)$. Moreover, this is possible because we are viewing the entire post-processing behavior as a new behavior in which we concatenate the post-processing behaviors from $\mathcal{W}_1$ and $\mathcal{W}_2$ into a single one, and use the flag $\Lambda$ to decide which one to use (similarly as we have done with the pre-processing). This is why we are allowed to write $p(\phi|\Lambda)$. This calculation is somewhat surprising because, effectively, it is possible to signal information from pre-processing randomness towards post-processing randomness with NCW operations, despite claims that this was not the case~\cite{amaral2018noncontextual,amaral2019resource,wolfe2020quantifying}. 

In this manner, we can write down $\mathcal{W}$ as
\begin{equation*}
    p_\beta(\mathbf{t}) = \sum_{\mathbf{r}',\mathbf{s}}p_{\beta}(\mathbf{r}')p_{\gamma=f(\mathbf{r}')}(\mathbf{s})p_{\delta=g(\mathbf{s}|\mathbf{r}',\beta)}(\mathbf{t}).
\end{equation*}
By construction, $\mathcal{W} = \alpha \mathcal{W}_1 + (1-\alpha)\mathcal{W}_2$. From this calculation, we conclude the following
\begin{theorem}\label{Theorem: NCW type-dependent convexity}
    Let $\Upsilon_{\mathrm{in}}, \Upsilon_{\mathrm{out}}$ be two compatibility scenarios. The set of all noncontextual wirings with a fixed type $\mathrm{NCW}(\Upsilon_{\mathrm{in}} \to \Upsilon_{\mathrm{out}})$ is a convex set.
\end{theorem}

The immediate corollary is the following one:

\begin{corollary}
    Let $\Upsilon_{\mathrm{in}}^{\mathrm{Bell}}, \Upsilon_{\mathrm{out}}^{\mathrm{Bell}}$ be two compatibility scenarios mathematically isomorphic to some Bell scenario. The set $\mathrm{NCW}(\Upsilon_{\mathrm{in}}^{\mathrm{Bell}} \to \Upsilon_{\mathrm{out}}^{\mathrm{Bell}})$ is a convex set.   
\end{corollary}

In particular, this corollary implies that every convex combination of elements in $\mathrm{LOSR}^\circ$ can be represented both as some noncontextual wiring and as some $\mathrm{LOSR}^\bullet$ operation. Effectively, as we have seen, what is allowed by a noncontextual wiring operation is that the classical randomness in the pre-processing can be fed forward freely toward the classical randomness of the post-processing boxes. This is operationally equivalent to allowing, in general, to have a single common source of shared randomness between pre- and post-processing boxes (in the case described above, that common source of randomness had the form $\sum_{\Lambda,\phi,\lambda}p(\Lambda)p(\phi|\Lambda)p(\lambda|\Lambda)$). 

In other words, the ability to prepare a source of randomness, and forward this information via some side channel and a function $f$ that post-process this information as we did above in the proof of Theorem~\ref{Theorem: NCW type-dependent convexity} is operationally equivalent to the ability to simply have a third side-channel that forward information of the classical source of randomness from the pre-processing behavior towards the post-processing box. We can therefore write any noncontextual wiring in the following form:
\begin{widetext}
\begin{equation}\label{eq: novel noncontextual wiring}
    p_\beta (\tb{t}) = \sum_\lambda p(\lambda)\sum_{\tb{r} \in \mathcal{O}^\beta,\tb{s}\in \mathcal{O}^\gamma} \prod_{\beta_i \in \beta} p_{\beta_i}(r_i|\lambda) p_{\gamma=f(\tb{r})} \left(\tb{s}\right) \prod_{\delta_i \in \delta = g(\mathbf{s}|\mathbf{r},\beta)}p_{\delta_i}(t_i|\lambda),
\end{equation}
\end{widetext}
where now $p(\lambda)$ is a common source of randomness for \emph{both} the pre- and post-processing behaviors. A similar remark was made in Ref.~\cite[Remark 28]{barbosa2023closing}. There, the authors point out the dual relationship between choosing to allow for a common source of shared randomness, or adaptivity in the free operations, that they interpret as simulations of different scenarios.

Therefore, we have also the following Corollary.

\begin{corollary}
    The following inclusions hold
    \begin{equation*}
        \mathrm{LOSR}^\circ \subsetneq \mathrm{LOSR}^\bullet \subsetneq \mathrm{NCW}|_{\mathrm{Bell}}.
    \end{equation*}
\end{corollary}

One result that will be instrumental for us, and that can be viewed as a corollary from the convexity of free operations, is the following.
\begin{corollary}\label{corollary: convexity}
    Let $B' \in \mathrm{NC}(\Upsilon)$, and $\alpha \in [0,1]$. For any resourceful behavior $B$ there exists a noncontextual wiring $\mathcal{W}$ such that 
    \begin{equation*}
        \mathcal{W}(B) = \alpha B + (1-\alpha)B'.
    \end{equation*}
\end{corollary}
\begin{proof}
    Consider the noncontextual wiring $\mathcal{W}_{B'}(B) = B'$ returning the noncontextual behavior $B'$ for any $B$ (as in Lemma~\ref{Lemma: noncontextual wiring always returning the same box}), and the noncontextual wiring that does nothing to the behaviors, i.e., $\mathcal{W}_{\mathrm{id}}(B)=B$. Now, consider the convex combination of these wirings via $\alpha$, that from Theorem~\ref{Theorem: NCW type-dependent convexity} describes a new noncontextual wiring. This wiring will be $\mathcal{W}$ that we wanted.
\end{proof}

Now, recalling the discussion on deterministic behaviors from Sec.~\ref{sec: noncontextual deterministic}, the convexity of our set of operations can be immediately shown to have more structure. If we recall that noncontextual behaviors can always be described, without loss of generality, as $\sum_\lambda p(\lambda) \prod_{\beta_i\in\beta} d_{\beta_i}(r_i|\lambda)$ and that we can write any noncontextual wiring acting on a behavior $\{p_\gamma(\mathbf{r})\}$ using Eq.~\eqref{eq: novel noncontextual wiring} we conclude that the action of any noncontextual wiring can be described as the convex combination of the action of some finite set of deterministic noncontextual wirings. The finite set of deterministic noncontextual wirings is determined from the size of $\lambda$. The set of all convex combinations of finitely many points precisely defines convex polytopes.

\begin{theorem}
    The free operations $\mathrm{NCW}(\Upsilon_{\mathrm{in}}\to\Upsilon_{\mathrm{out}})$ form a convex polytope.
\end{theorem}

The vertices of the convex polytope $\mathrm{NCW}(\Upsilon_{\mathrm{in}}\to\Upsilon_{\mathrm{out}})$ are given by deterministic noncontextual wirings. When restricted to Bell scenarios, because $\mathrm{LOSR}^\bullet$ is also a convex polytope, we have that the inclusion $\mathrm{LOSR}^\bullet \subsetneq \mathrm{NCW}|_{\mathrm{Bell}}$ when restricted to specific types of operations is a convex polytope inclusion.

\subsection{Global comparability properties}
\label{sec: global}

In this section, we investigate the properties from Def.~\ref{def: global properties}. We will prove the following Theorems.

\begin{theorem}\label{Theorem: existence of incomparable objects}
    The pre-order on objects induced by $\mathrm{NCW}$ is not a total pre-order, i.e., there exist pairs of behaviors that are incomparable under noncontextual wirings. 
\end{theorem}

\begin{theorem}\label{Theorem: weakness}
    The pre-order on objects induced by $\mathrm{NCW}$ is not weak, i.e., there exists triplets of distinct objects $B_1, B_2, B_3$ such that $B_1 \nleftrightarrow B_2, B_2 \nleftrightarrow B_3$ while $B_1$ and $B_3$ are comparable.
\end{theorem}

\begin{theorem}\label{Theorem: height and width}
    The height and the width of the pre-order on objects induced by $\mathrm{NCW}$ are both (uncountably) infinite.
\end{theorem}

\begin{theorem}\label{Theorem: local infiniteness}
    The pre-order described by $\mathrm{NCW}$ is locally infinite, i.e., there exists an interval $B_1 \to B \to B_2$ for which the cardinality of the set of all equivalence classes $[B] := \{B' : B \to B' \text{ and }B' \to B\}$ is (uncountably) infinite.
\end{theorem}

The way we show these results is by focusing on specific constructions on the $n$-cycle scenarios and by using noncontextuality monotones. As already mentioned in the introduction, Theorems~\ref{Theorem: existence of incomparable objects}-\ref{Theorem: local infiniteness} were known to hold in the case of the bipartite Bell scenario introduced by Clauser, Horne, Shimony, and Holt (CHSH)~\cite{clauser1969proposed} and when the free operations are taken to be $\mathrm{LOSR}^\bullet$~\cite{wolfe2020quantifying}. In what follows we will generalize these findings to hold for any $n$-cycle compatibility scenario $\Upsilon_n$ and with respect to free operations taken to be NCW. 

The techniques we will employ are straightforward translations from the tools discussed in Ref.~\cite{wolfe2020quantifying} for (bipartite) Bell scenarios to the case of $n$-cycle compatibility scenarios. For a description of the monotones we consider for any abstract resource theory, we refer the reader to Ref.~\cite{gonda2019monotones}. See also Ref.~\cite[Appendix C]{yile2024conceptualformalgroundworkstudy} for some results on global comparability properties (specifically width and weakness) in general resource theories. We will start by introducing the relevant cost and yield resource monotones. Then, we use these constructions to prove each of the Theorems listed above.

\subsubsection{Yield and cost monotones for $n$-cycle compatibility scenarios}

To prove our results for this section we will use two resource monotones known as yield $\mathsf{y}_k$ and cost $\mathsf{c}_k$. These monotones are defined relative to specific (fixed) functionals, that we take here to be given by the functionals $I_k^{(n)}$ defining facet-defining inequalities of the noncontextual polytope for $n$-cycle scenarios, given by Eq.~\eqref{eq: noncontextuality inequalities}. Moreover, these monotones are also defined relative to the specific free operations in our resource theory. 

We define $\mathsf{y}_k: \mathrm{ND}(\Upsilon_n) \to \mathbb{R}$ to be the yield monotone defined over any $n$-cycle scenario $\Upsilon_n$, with $n\geq 3$, as
\begin{equation}
    \mathsf{y}_k(B) := \max_{B' \in \mathrm{ND}(\Upsilon_n)}\{I_k^{(n)}(B'): B \to B'\},
\end{equation}
    where $\{I_k^{(n)}\}_k$ are all facet-defining noncontextuality inequality-functionals of $\mathrm{NC}(\Upsilon_n)$, given by Eq.~\eqref{eq: noncontextuality inequalities}. Above, the arrow $B \to B'$ represents that there exists some $\mathcal{W} \in \mathrm{NCW}(\Upsilon_n \to \Upsilon_n)$, for any fixed $n\geq 3$, such that $B' = \mathcal{W}(B)$. The yield $\mathsf{y}_k$ gives the value of $I_k^{(n)}$ for the most resourceful behavior that can be freely obtained from the behavior $B$. Note also that for every $B$ there exists some $B'$ such that $B \to B'$, hence the value of this monotone is always bounded from below by $n-2$. For any $k$, and any noncontextual wiring $\mathcal{W}\in \mathrm{NCW}(\Upsilon_n \to \Upsilon_n)$, 
    
    \begin{align*}
        \mathsf{y}_k(\mathcal{W}(B)) &= \max_{B'}\{I_k^{(n)}(B'):\mathcal{W}(B) \to B'\}\\
        &=\max_{B'}\{I_k^{(n)}(B')|B' = \mathcal{W}'\circ \mathcal{W}(B)\text{ for some }\mathcal{W}'\}\\
        &\leq \max_{B'}\{I_k^{(n)}(B')|B' = \mathcal{W}'(B)\text{ for some }\mathcal{W}'\}\\
        &= \max_{B'}\{I_k^{(n)}(B')|B \to B'\}\\
        &=\mathsf{y}_k(B).
    \end{align*}
Where the inequality comes from the fact that, for the first optimization, the possible free operations must have the decomposition $\mathcal{W}'\circ \mathcal{W}$, with $\mathcal{W}$ fixed, while the second is left to be any free operation. These calculations show that $\mathsf{y}_k$ are indeed resource monotones.

\begin{lemma}[Yield monotone for $n$-cycle scenarios]\label{Lemma: yield}
 For all functionals $I_k^{(n)}$ and all $B \in \mathrm{NC}(\Upsilon_n)$, we have that  $$\mathsf{y}_k(B)=n-2.$$ Moreover, for any $B \in \mathrm{ND}(\Upsilon_n)\setminus \mathrm{NC}(\Upsilon_n)$ there exists some unique label $k^\star$ such that $$\mathsf{y}_{k^\star}(B) = I_{k^\star}^{(n)}(B).$$ 
\end{lemma}

\begin{proof}
    The second part is trivial, following from the definition and from the fact that any contextual behavior in $\Upsilon_n$, for any $n\geq 3$, violates one and only one noncontextuality inequality. The first part follows from the fact that all free resources are equivalent (see Lemma~\ref{Lemma: free equivalent}), implying that one can simply take $B'$ to be a noncontextual behavior for which $I_k^{(n)}(B') = n-2$. 
\end{proof}

We will also need another monotone, termed the cost monotone and denoted as  $\mathsf{c}_k$. However, to define this monotone precisely we will need to introduce new sets of behaviors. For that, we must recall some basic aspects of the polytopes $\mathrm{ND}(\Upsilon_n)$ and $\mathrm{NC}(\Upsilon_n)$ defined for $n$-cycle compatibility scenarios. We have already mentioned that every contextual behavior in such scenarios violates one, and only one, facet-defining noncontextuality inequality.  Recall that, if we denote all the vertices of a convex polytope $P$ as $\text{ext}(P)$, we have that  $\text{ext}(\mathrm{NC}(\Upsilon_n)) \subseteq \text{ext}(\mathrm{ND}(\Upsilon_n))$. In particular, it is therefore also true that each Ineq.~\eqref{eq: noncontextuality inequalities} is violated by one, and only one vertex in $\text{ext}(\mathrm{ND}(\Upsilon_n)) \setminus \text{ext}(\mathrm{NC}(\Upsilon_n))$. Vertices in $\text{ext}(\mathrm{ND}(\Upsilon_n)) \setminus \text{ext}(\mathrm{NC}(\Upsilon_n))$ are always \emph{strongly contextual}~\cite{abramsky2017contextual} and are never \emph{quantum realizable} (see Ref.~\cite{fraser2023estimationtheoreticapproachquantum} for an introduction to quantum realizability problems), i.e., these behaviors cannot be quantum behaviors. 

Let us denote by $$\mathrm{NDC}(\Upsilon_n):= \text{ext}(\mathrm{ND}(\Upsilon_n)) \setminus \text{ext}(\mathrm{NC}(\Upsilon_n))$$
the set of all $2^{n-1}$ non-disturbing contextual vertices~\cite{araujo2013all}. Let us also denote the specific behavior $B_{\varnothing}$ where all outcomes are equally likely, implying that all two-point correlation functions equal zero, and therefore $I_k^{(n)}(B_\varnothing) = 0$. We now define a discrete set, given by all behaviors of the form
\begin{equation*}    
\mathrm{D}_n:= \left\{\frac{n-2}{n}B+\frac{2}{n}B_\varnothing \,:\,  B \in \mathrm{NDC}(\Upsilon_n)\right \}
\end{equation*}
where, of course, $B_\varnothing \in \mathrm{NC}(\Upsilon_n)$. By construction, $\mathrm{D}_n \subseteq \mathrm{NC}(\Upsilon_n)$ as will soon be clear.

For every  $B \in \mathrm{D}_n$ we have that there exists a specific functional $I_k^{(n)}$ that has value $I_k^{(n)}(B) = n-2$. This is simply because $I_k^{(n)}(B_\varnothing) = 0$ and the remaining behavior in the combination returns $n$ since it saturates the algebraic maximum of the functional  $I_k^{(n)}$. Because of that, these points are exactly those that lie in the intersection between the line from $B_\varnothing$ to some extremal contextual behavior, and the facet of all behaviors returning a tight value to the specific inequality $I_k^{(n)}$. We can therefore introduce a continuum of behaviors that lie within these specific points in $\mathrm{D}_n$ and the elements of $\mathrm{NDC}(\Upsilon_n)$ 

\begin{align*}
    \mathrm{P}_n&:= \bigsqcup_k \mathrm{P}_k^{(n)},
\end{align*}
where $\mathrm{P}_n$ is the disjoint union of sets $$\mathrm{P}_k^{(n)}= \{\varepsilon B_k + (1-\varepsilon)B_k'|\varepsilon \in [0,1]\},$$ where $B_k$ is the unique element from $\mathrm{NDC}(\Upsilon_n)$ reaching the algebraic maximum of $I_k^{(n)}$, and $B_k' := (n-2)/n B_k+2/n B_\varnothing$. 

\begin{figure}[t]
    \centering
    \includegraphics[width=\columnwidth]{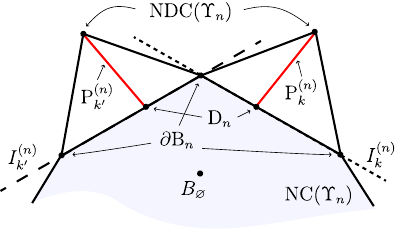}
    \caption{Convex polytopes $\mathrm{ND}(\Upsilon_n)$ and $\mathrm{NC}(\Upsilon_n)$, together with relevant sets of behaviors for this section, for a fixed $n$.  The sets of behaviors $\mathrm{NDC}(\Upsilon_n), \mathrm{D}_n, \mathrm{P}_k^{(n)}, \partial \mathrm{B}_n$ described in the text are shown. Each inequality functional $I_k^{(n)}$ or $I_{k'}^{(n)}$ defines a hyperplane via $I_k^{(n)}(B) = n-2$. The points in $\mathrm{D}_n$ lie in the intersection of these hyperplanes with the line between $B_{\varnothing}$ and elements in $\mathrm{NDC}(\Upsilon_n)$ that are strongly contextual behaviors. For each $k$ labeling a facet-defining noncontextuality inequality the sets $\mathrm{P}_k^{(n)}$ are defined by the convex combinations of the specific points in $\mathrm{D}_n$ and in $\mathrm{NDC}(\Upsilon_n)$  defined via the label $k$.}
    \label{fig: sets of behaviors}
\end{figure}

The last set of behaviors we need to introduce is the set of behaviors that both saturate the noncontextuality inequalities, and hence lie in the boundary of $\mathrm{NC}(\Upsilon_n)$, and also lie in the boundary of $\mathrm{ND}(\Upsilon_n)$. 
 For each $k$ labeling the facet-defining inequalities, if we consider the set $\{B \in \mathrm{ext}(\mathrm{NC}(\Upsilon_n)):I_k^{(n)}=n-2\}$ these sets are  (obviously) non-empty and have cardinality greater than $1$. Let $B^\star_k$ be one such element, for each $k$. We define the set $\partial\mathrm{B}_n := \cup_k\{B_k \in \text{ext}(\mathrm{NC}(\Upsilon_n))|I_k^{(n)}(B_k)=n-2\}$, to be the set of all such choices. In words, each element of $\partial \mathrm{B}_n$ is an extremal element of the noncontextual polytope that saturates at least one of the facet-defining noncontextuality inequalities~\ref{eq: noncontextuality inequalities}. All the relevant sets of behaviors just described are shown in Fig.~\ref{fig: sets of behaviors}.

With the technical ingredients just discussed, we can now define our cost monotone. Let $\mathsf{c}_k: \mathrm{ND}(\Upsilon_n) \to \mathbb{R}\cup \{+\infty,-\infty\}$ be the cost monotone defined  over any $n$-cycle scenario $\Upsilon_n$, with $n\geq 3$, as 
    \begin{equation}
        \mathsf{c}_k(B):= \min_{B' \in \mathrm{P}_k^{(n)}}\{I_k^{(n)}(B'): B' \to B\},
    \end{equation}
    where $\{I_k^{(n)}\}_k$ are all facet-defining noncontextuality inequality-functionals of $\mathrm{NC}(\Upsilon_n)$ from Eq.~\eqref{eq: noncontextuality inequalities}. Following a similar reasoning to the yield $\mathsf{y}_k$ it can be shown that $\mathsf{c}_k$ is also a resource monotone. Note, however, that it is not always the case that for a given behavior $B$ there will be some $B' \in \mathrm{P}_n$ such that $B' \to B$. This \emph{is} the case whenever $B \in \mathrm{NC}(\Upsilon_n)$ or $B$ violates the same noncontextuality inequality as $B'$. However, if $B$ violates some other inequality, then there will be \emph{no} free operation such that $\mathrm{P}_k^{(n)} \ni B' \to B$, and in such cases we set the value $\mathsf{c}_k(B) = +\infty$. This situation corresponds in Fig.~\ref{fig: sets of behaviors} for behaviors $B, B'$ in different regions of contextual behaviors, represented by the two triangles.

As before with the yield monotone, we can show the following:
    
\begin{lemma}[Cost monotone for $n$-cycle scenarios]\label{Lemma: cost}
    For any $I_k^{(n)}$ and any $B \in \mathrm{NC}(\Upsilon_n)$ we have that $$\mathsf{c}_k(B) = n-2.$$ Moreover, for any $B \in \mathrm{P}_n\setminus \mathrm{NC}(\Upsilon_n)$ there exists $k^\star$ such that $$\mathsf{c}_{k^\star}(B) = I_{k^\star}^{(n)}(B).$$ More generally, for any $B \in \mathrm{ND}(\Upsilon_n)$ we have that there exists $k^\star$ and $\varepsilon \in [0,1]$ such that $$\mathsf{c}_{k^\star}(B) = n+2(\varepsilon-1),$$
    where $\varepsilon$ is given by 
    \begin{equation}\label{eq: definition B(e,a)}
        B = B(\alpha,\varepsilon) := \alpha \tilde{B} + (1-\alpha) B_\varepsilon.
    \end{equation}
    Above, $\alpha \in [0,1]$ is some value returned by the optimization and guaranteed to exist, while  $B_\varepsilon \in \mathrm{P}_n$ and  $\tilde{B} = B_{k^\star} \in \partial \mathrm{B}_n$ with $I_{k^\star}(B_{k^\star})=n-2$. 
\end{lemma}

We prove this result in the Appendix~\ref{app: proof of lemma cost}. Above, for each $B$ there exists at least one $B_{k^\star} \in \partial \mathrm{B}_n$ describing the decomposition. Using these two lemmas, it is simple to show that the following lemma holds
\begin{lemma}[Cost and yield for the family $B(\alpha,\varepsilon)$]\label{Lemma: cost and yield for B(e,a)}
    Let $B(\alpha,\varepsilon)$ be a behavior from a scenario $\Upsilon_n$ given as in Eq.~\eqref{eq: definition B(e,a)}, violating a  noncontextuality inequality $I_k^{(n)}(B) \leq n-2$, with a fixed $k$. Then, 
    \begin{equation}\label{eq: cost of B(a,e)}        \mathsf{c}_k(B(\alpha,\varepsilon)) = n+2(\varepsilon-1),
    \end{equation}
    and 
    \begin{equation}\label{eq: yield of B(a,e)}
        \mathsf{y}_k(B(\alpha,\varepsilon))=n-2+2\varepsilon(1-\alpha),
    \end{equation}
    for any $\varepsilon,\alpha \in [0,1]$.
\end{lemma}

\begin{proof}
    The form of the cost follows trivially from Lemma~\ref{Lemma: cost}, by construction. The form of the yield follows from the fact that any such behavior is given as a combination of the form $$\alpha \tilde{B}+(1-\alpha)\varepsilon B + (1-\alpha)(1-\varepsilon)B',$$ where, if contextual, the yield returns the value of the inequality functional that is linear and then we simply need to calculate it for $\tilde{B}, B$, and $B'$. But by construction, $I_k^{(n)}(\tilde{B})=I_k^{(n)}(B')=n-2$ and $I_k^{(n)}(B) =n$ since $\tilde{B} \in \partial \mathrm{B}_n$, $B' \in \mathrm{D}_n$ (both chosen such that $I_k^{(n)}(\tilde{B}) = n-2$) and $B \in \mathrm{NDC}(\Upsilon_n)$ is a non-disturbing contextual vertex of $\mathrm{ND}(\Upsilon_n)$. These imply that,
    \begin{align*}
        \mathsf{y}_k(B(\alpha,\varepsilon)) &= \alpha (n-2)+(1-\alpha)\{(1-\varepsilon)(n-2)+\varepsilon n\}\\
        &=n-2+2\varepsilon(1-\alpha).
    \end{align*}
\end{proof}

Note that the lemma above holds once we fixed some set of contextual behaviors violating a specific inequality $I_k^{(n)}(B) \leq n-2$, represented the triangles in Fig.~\ref{fig: sets of behaviors}.  With these two resource monotones and the specific form they have when applied to $n$-cycle scenarios, we now proceed to prove Theorems~\ref{Theorem: existence of incomparable objects}-\ref{Theorem: local infiniteness}.

\subsubsection{Incomparable behaviors under wirings}

We start showing that there are incomparable objects under noncontextual wirings, proving Theorem~\ref{Theorem: existence of incomparable objects}. To prove this we can use resource monotones. If we have two inequivalent monotones $\mathsf{m}_1, \mathsf{m}_2$ and two objects $B_1,B_2$ such that $\mathsf{m}_1(B_1)<\mathsf{m}_1(B_2)$ and $\mathsf{m}_2(B_2)<\mathsf{m}_2(B_1)$ we get that both objects must be incomparable. This is because, as we have seen in Sec.~\ref{subsec: resource monotones}, the first inequality implies that $B_1$ cannot be freely transformed into $B_2$, while the second inequality implies that $B_2$ cannot be freely transformed into $B_1$. 

For this, consider the following two objects: $B_1 = B(1/4,1/4)$ and $B_2=B(3/4,1/2)$, where we are considering a parametrized family $B(\alpha,\varepsilon)$ given by Eq.~\eqref{eq: definition B(e,a)}. From Lemma~\ref{Lemma: cost and yield for B(e,a)} we have that
\begin{equation*}
    \mathsf{c}_k(B_1)=n-\frac{3}{2}<n-1=\mathsf{c}_k(B_2)
\end{equation*}
and
\begin{equation*}
    \mathsf{y}_k(B_1)=n-\frac{13}{8}>n-\frac{7}{4}=\mathsf{y}_k(B_2).
\end{equation*}

We conclude that \emph{the pre-order is not totally pre-ordered}, proving Theorem~\ref{Theorem: existence of incomparable objects}. Note that this incomparability result remains true for all $n$-cycle scenarios. 

\subsubsection{Weakness of incomparability}

Given that we have incomparable behaviors, we can study the properties of the incomparability relation. Now, consider $B_1$ and $B_2$ as in the last section, and define a new behavior given by $B_3 = B(1/8,1/4)$. Since we have that $B_3$ and $B_1$ have the same value of $\varepsilon$ these are both convex combinations of a free behavior $\tilde{B} \in \partial \mathrm{B}_n \subseteq \mathrm{NC}(\Upsilon_n)$ and a (resourceful) behavior $B_{\varepsilon=1/4} \in \mathrm{P}_n$. In the last section, we have seen that $B_1$ and $B_2$ are incomparable. Following the same reasoning, we can also see that $B_2$ and $B_3$ are incomparable, because
\begin{equation*}
    \mathsf{c}_k(B_3) = \mathsf{c}_k(B_1) = n-\frac{3}{2}<n-1=\mathsf{c}_k(B_2)
\end{equation*}
and
\begin{equation*}
    \mathsf{y}_k(B_3)=n-\frac{25}{16}>n-\frac{7}{4}=\mathsf{y}_k(B_2).
\end{equation*}
Therefore $B_1 \nleftrightarrow B_2$ and $B_2 \nleftrightarrow B_3$. Since $\tilde{B}$ is a free behavior, and $B_1$ can be written as a convex combination of $B_3$ and $\tilde{B}$, there exists a noncontextual wiring $B_3 \to B_1$ (see Corollary~\ref{corollary: convexity}). As a sanity check, we can calculate both the cost and yield, and see that for both, the values for $B_3$ are higher. In fact,  $\mathsf{c}_k(B_1) = \mathsf{c}_k(B_3)$ and $\mathsf{y}_k(B_1)<\mathsf{y}_k(B_3)$, as expected. 

These calculations show that incomparability between objects in our resource theory, denoted as  $\nleftrightarrow$, is not a transitive relation. This implies that the pre-order of noncontextual wirings is not weak, and concludes the proof of Theorem~\ref{Theorem: weakness}.

\subsubsection{Height and width}

Recall that the \emph{height} of the pre-order is the cardinality of the largest chain contained therein (see Def.~\ref{def: global properties}). A finite chain would be for instance a chain of the form $B_1 \to B_2 \to B_3$ and a countable but infinite chain would be a chain of the form $B_1 \to B_2 \to B_3 \to \dots $ that has the same cardinality as the one of the natural numbers. An \emph{uncountable} chain is one that, for any given interval of real numbers $I=[a,b]$ we have that $\forall x_1,x_2 \in I$, $x_1 \leq x_2$ implies that $B_{x_1} \to B_{x_2}$. 

Note that we know that any element of  $\mathrm{P}_n$ is a convex combination (defined by a parameter $\varepsilon \in [0,1]$) of some resourceful behavior $B$ and a free behavior $B'$. Given any such free and resourceful behaviors, noncontextual wirings generate any possible convex combination between $B$ and $B'$, from Corollary~\ref{corollary: convexity}. This implies that for any pair $\varepsilon_1,\varepsilon_2 \in [0,1]$ such that $\varepsilon_1 \leq \varepsilon_2$ we will have that $B_{1-\varepsilon_1} \to B_{1-\varepsilon_2}$, and this forms an uncountable infinite chain. Therefore we conclude that the height of the pre-order is (uncountably) infinite, and has the same cardinality as the set of real numbers. 

To investigate the width of the pre-order, consider the set of points $\left\{ B(x,x) \mid \frac{1}{2} \leq x \leq 1 \right\},$ the line segment between points $B\left(\frac{1}{2}, \frac{1}{2}\right)$ and $B\left(1,1\right)$. By inspection, we notice that within this region, the function $\mathsf{c}_k(B(x,x)) = n+2(x-1)$ is strictly \emph{increasing}, while $\mathsf{y}_k(B(x,x))=n-2+2x(1-x)$ is strictly \emph{decreasing}. Therefore, this pair of monotones witness the incomparability of every pair of objects in this line segment. Hence this line segment constitutes an antichain, and since here there is also an (uncountably) infinite number of incomparable points, by the same logic applied to the height of the pre-order, we conclude that the width of the pre-order is also (uncountably) infinite. This concludes the proof of Theorem~\ref{Theorem: height and width}.

\subsubsection{Local finiteness}

To conclude, we show the claim from Theorem~\ref{Theorem: local infiniteness}. We want to show that there exists an infinite set of inequivalent behaviors within some interval chain. To show that, the infinite chain $\mathrm{P}_n$ is also instrumental. Note that for each interval $B_{1-\varepsilon_1} \to B_{1-\varepsilon} \to B_{1-\varepsilon_2}$ there exists an infinite set of inequivalent behaviors. This is captured, for instance, by the cost monotone $\mathsf{c}_k$ acting on $B(0,\varepsilon)$. For any $\varepsilon_1 < \varepsilon < \varepsilon_2$ we have that $\mathsf{c}_k(B_{1-\varepsilon_1})<\mathsf{c}_k(B_{1-\varepsilon})<\mathsf{c}_k(B_{1-\varepsilon_2})$ and there can exist no free operation from $B_{1-\varepsilon_2} \to B_{1-\varepsilon}$ and neither from $B_{1-\varepsilon} \to B_{1-\varepsilon_1}$. This implies that all these are inequivalent resources, and since this is true for the interval, there are infinitely many classes of inequivalent resources within $\mathrm{P}_k^{(n)}$, for any $k$ and $n$. This shows that the pre-order induced by noncontextual wirings over behaviors in any $n$-cycle scenario is locally infinite (see Def.~\ref{def: global properties}). This concludes the proof of Theorem~\ref{Theorem: local infiniteness}.

As a final remark, we notice that due to Vorobyev's Theorem~\cite{vorobyev1962consistent}, any measurement scenario $\Upsilon$ for which $\mathrm{NC}(\Upsilon)\neq \mathrm{ND}(\Upsilon)$, sometimes called \emph{contextuality-witnessing} scenarios, must contain an $n$-cycle sub scenario. Via lifting~\cite{choudhary2024lifting} any contextuality-witnessing scenario must therefore contain lifted $n$-cycle inequalities, for which one can apply our findings by generalizing straightforwardly the cost and yield monotones we have described, and conclude that all Theorems~\ref{Theorem: existence of incomparable objects}-\ref{Theorem: local infiniteness} hold in any scenario, not only dichotomic cycle scenarios. 

\section{Discussion}\label{sec: discussion}

In this work we have studied the resource theory one obtains when considering behaviors in compatibility scenarios as resource objects, Kochen--Specker noncontextual behaviors as free objects, and noncontextual wirings as free operations. While these free operations have been introduced in previous literature, we believe our presentation of the topic is significantly different than those present therein, and it clarifies some aspects that were either known as folklore or overlooked. One such aspect is the relationship between the set of noncontextual wirings, when applied to compatibility scenarios isomorphic to Bell scenarios, and the set of local operations and shared randomness in such scenarios. We show that noncontextual wirings are not equivalent to local operations with shared randomness, but contain these operations as a subset of possible operations.

One operational aspect we show is the \emph{convexity} of the set of noncontextual wirings. We prove that this set is closed under convex combinations. In proving so, we end up showing that the ability to feedforward input/outputs to future behaviors (side-channels in a wiring) \emph{implies} that both pre- and post-processing noncontextual behaviors can share the same source of randomness. This aspect led to a novel description of how one may define noncontextual wirings, allowing for a common source of shared randomness between input/output behaviors. Clearly, this discussion, in particular, is motivated by the recent findings of Wolfe \emph{et.~al.}~\cite{wolfe2020quantifying} that recognized the important role of having a common source of shared randomness between pre- and post-processing behaviors when discussing $\mathrm{LOSR}^\bullet$. Our results showed that, as with the set of type-dependent local operations with shared randomness, the set of type-dependent noncontextual wirings form a polytope. 

Moreover, motivated by the study of global comparability properties of the resource theory of nonclassical common-cause boxes (Def.~\ref{def: global properties}) by Wolfe \emph{et. al.}~\cite{wolfe2020quantifying}, we have also investigated which ordering properties noncontextual wirings induce over the set of all behaviors.  To do so, we use cost and yield monotones, specifically targeting facets of $n$-cycle compatibility scenarios, to show (i) the existence of incomparable objects under wirings, (ii) that the pre-order is not weak, (iii) that both the height and width of the pre-order are infinite and finally, (iv) there exists an interval having infinitely many inequivalent objects. Along the way to showing these results, we also show elementary facts, such as that every noncontextual behavior is equivalent under noncontextual wirings, that convex mixtures are a subset of all possible wirings, and that any resourceful object can be freely operated towards any free object.

\subsection{Relation with previous work}

Our work is significantly motivated by the findings of Wolfe \emph{et. al.}~\cite{wolfe2020quantifying}. Specifically, we were interested in generalizing the generic features of the global comparability properties they find in the context of the CHSH Bell scenario to any contextuality-witnessing $n$-cycle scenario. Because, as we have shown, noncontextual wirings are a larger set of free operations (within compatibility scenarios) than local operations and shared randomness, it could be that incomparable objects under LOSR become \emph{comparable} under NCW, inflicting in different global properties of the pre-order. Perhaps unsurprisingly, using fairly similar tools as the one employed by Ref.~\cite{wolfe2020quantifying}, we have obtained the same global properties. 

One major aspect of our work is, however, understanding the convexity of free operations NCW. Much of the discussion from Ref.~\cite{wolfe2020quantifying} suggests (in our view, correctly) that unless there exists a mechanism to forward randomness from the pre-processing behaviors towards the post-processing ones, there will be an inevitable non-convexity of the set of operations. Intriguingly, previous literature has explicitly demanded that noncontextual wirings could not forward information of the states $\lambda$. This choice was made specifically because the ability to do so could lead to the creation of contextuality using mechanisms akin to those found in Refs.~\cite{budroni2019contextuality,kleinmann2011memory,fagundes2017memory,budroni2019memory_temporal}. Our work, despite what is suggested by the discussion of Ref.~\cite[Appendix A]{wolfe2020quantifying}, shows that the set of type-dependent noncontextual wirings is a convex polytope, and therefore a convex set. For this conclusion to hold we use the fact that \emph{it is possible} to feedforward information of the states $\lambda$ in the pre-processing towards the post-processing, and that this aspect is equivalent to allowing a common source of randomness for both the pre- and post-processing behaviors, as it happens with $\mathrm{LOSR}^\bullet$. In particular, this also shows that, because noncontextual wirings are free, they cannot create contextuality via the memory mechanisms of Refs.~\cite{budroni2019contextuality,kleinmann2011memory,fagundes2017memory,budroni2019memory_temporal} as initially thought. 

Two connections with existing literature are worth mentioning. First, while references~\cite{amaral2018noncontextual,amaral2019resource} claimed that local operations and shared randomness were equivalent to noncontextual wirings restricted to Bell scenarios, reference~\cite{karvonen2021neither} briefly pointed out that the set of all possible free operations for the resource theory of contextuality needed to be larger than that of the LOSR when restricted to Bell scenarios. Moreover, Ref.~\cite{barbosa2023closing} remarked (see Remark 28 therein) that one can exchange adaptivity and a common source of randomness. Our work brings clarity to the discussions present in the references mentioned above regarding the convexity of NCW and its relation with LOSR.

\subsection{Further directions}

Various aspects are left for future work. For example, as mentioned above, Refs.~\cite{karvonen2021neither,abramsky2019comonadic,abramsky2017contextual,barbosa2023closing} considered a larger set of possible free operations than that of noncontextual wirings. It remains to see what global comparability properties these free operations have. To the best of our knowledge, it is unknown if there are incomparable behaviors for their notion of free operations. 

Another interesting aspect is to understand how generic wirings (not just noncontextual ones) behave. For instance, there are likely classes of contextual wirings that combined can either send a resourceful behavior towards a free behavior, or that act trivially on the resources. The mapping of such possibilities would guarantee that one does not \emph{wastes} resources by acting with resourceful wirings without gaining anything from these operations. In other words, denoting as $\mathcal{W}^{\mathsf{c}}$ a non-free wiring that has contextual pre- or post-processing behaviors, is it possible that $\mathcal{W}^{\mathsf{c}}(B) = \mathcal{W}(B^\star)$, for some free wiring $\mathcal{W}$ and some other behavior $B^\star$? And if so, can this effect be completely characterized, with necessary and sufficient conditions deciding when this is possible? One other possibility is also to investigate `doped' resources~\cite{peres2024nonstabilizerness,haferkamp2022efficient}, when we have a sequence of free transformations $$\mathcal{W}_1 \circ \mathcal{W}_2 \circ \dots \circ \mathcal{W}^{\mathsf{cost}} \circ \dots \circ \mathcal{W}_n(B)$$ `doped' with few non-free wirings $\mathcal{W}^{\mathsf{c}}$. Understanding the quantum information power of these restricted classes of operations may be useful for resource quantification.


\begin{acknowledgments}
We would like to thank Ana Bel\'en Sainz, Rui Soares Barbosa, Laurens Walleghem and Costantino Budroni for  helpful comments and discussions. TW and BA would like to acknowledge financial support from Conselho Nacional de Desenvolvimento Cient\'ifico e Tecnol\'ogico (CNPQ) - Regular program, Grant number $131630/2019-9$, CNPq, Chamada Universal 2018, Grant number $431443/2018-1$, and from Fundação de Amparo \`a Pesquisa do Estado de S\~ao Paulo, Aux\'ilio \`a Pesquisa - Jovem Pesquisador, grant number $2020/06454-7$, and Instituto Serrapilheira, Chamada 2020. RW would like to acknowledge financial support from FCT – Fundação para a Ciência e a Tecnologia (Portugal) through PhD Grant SFRH/BD/151199/2021. This work was also supported by the Digital Horizon Europe project FoQaCiA, GA no.101070558, funded by the European Union, NSERC (Canada), and UKRI (U.K.).
\end{acknowledgments}
 
\bibliography{biblio}

\begin{appendix}

\section{Proof of Lemma~\ref{Lemma: cost}}\label{app: proof of lemma cost}

Let us suppose that $B \in \mathrm{NC}(\Upsilon_n)$. Then, by definition of the monotone we have that 
\begin{align*}
\mathsf{c}_k(B) &= \min_{B' \in \mathrm{P}_n} \{ I_k^{(n)}(B'): B' \to B \} = n-2,
\end{align*}
since, because $B$ is a noncontextual behavior there exists some wiring sending every $B' \in \mathrm{P}_n$ towards the only noncontextual element $B \in \mathrm{P}_n$ that saturates some noncontextuality inequality $I_k^{(n)} = n-2$. From the definition, if $B \in \mathrm{P}_n \setminus \mathrm{NC}(\Upsilon_n)$ we have that there exists some $k^\star$ such that $\mathsf{c}_{k^\star}(B) = I_{k^\star}^{(n)}(B)$.

Let us consider more generally any $B \in \mathrm{ND}(\Upsilon_n)$. If $B \in \mathrm{ND}(\Upsilon_n) \setminus \mathrm{NC}(\Upsilon_n)$, let $k^\star$ be the label for the inequality functional $I_{k^\star}^{(n)}$ that is violated. $B$ can always be written as $$B = B(\alpha,\varepsilon) := \alpha \tilde{B} + (1-\alpha) B_\varepsilon$$
for $\tilde{B}$ some extremal noncontextual behavior satisfying that $I_{k^\star}^{(n)} = n-2$, and $B_\varepsilon \in \mathrm{P}_{k^\star}^{(n)}$.

Since $B$ is a convex combination of $\tilde{B}$ (which is a free behavior) and $B_\varepsilon$, there exists a noncontextual wiring such that $B_\varepsilon \mapsto B$. Clearly, for any $\varepsilon' > \varepsilon$ we increase the value of $I_{k^\star}^{(n)}(B'_{\varepsilon'})$, and hence make the minimization defining $\mathsf{c}_{k^\star}$ worse. For values $\varepsilon' < \varepsilon$ the value $I_{k^\star}^{(n)}(B'_{\varepsilon'})$ decreases, but there can be no free operation towards $B$. One can see this by supposing that there exists a free operation from $B_{\varepsilon'} \to B = \alpha \tilde{B}+(1-\alpha)B_\varepsilon \equiv B({\alpha,\varepsilon})$ with $\varepsilon' < \varepsilon$. Clearly, there is a free operation from $B \to B(0,\varepsilon) = \tilde{B}$, due to the convexity of NCW. Therefore, we would have that $B_{\varepsilon'} \to B(\alpha,\varepsilon) \to B(0,\varepsilon)$. From transitivity, we would have that there exists a free operation from $B_{\varepsilon'}$ to $B_\varepsilon$, which is absurd. 

Therefore, we can simply use
\begin{align*}
    &\mathsf{c}_{k^\star}(B) = I_{k^\star}(B_\varepsilon) = \varepsilon n + (1-\varepsilon)(n-2)\\
    &=\varepsilon n + n-2 -\varepsilon n+2\varepsilon\\
    &=n+2(\varepsilon-1),
\end{align*}
as we wanted to show.

\end{appendix}

\end{document}